\documentclass[journal]{IEEEtran}
\usepackage{amsmath,amssymb,amsfonts}
\usepackage{bm}
\usepackage{graphicx}
\usepackage{cite}
\usepackage{algorithm}
\usepackage{algorithmic}
\graphicspath{{figures/}}

\newtheorem{theorem}{Theorem}
\newtheorem{lemma}{Lemma}
\newtheorem{proposition}{Proposition}

\newtheorem{remark}{Remark}
\newtheorem{assumption}{Assumption}

\begin{document}

\title{Deterministic DTFT Interpolation for Joint Frequency and
Chirp-Rate Estimation: Cell-Uniform Efficiency and Threshold Analysis}

\author{Miaomiao~Wei, Jianjun~Li, Yang~Wang, Huaiyuan~Chen, Lulu~Gao,
and~Hang~Liu%
\thanks{The authors are with Zhongyuan University of Technology, Zhengzhou,
China (corresponding author: Miaomiao Wei, e-mail: mmwei@zut.edu.cn).}%
\thanks{This work was supported in part by the National Natural Science
Foundation of China under Grant 62301624, and in part by the Key Science
and Technology Research Project of Henan Province under Grant
242102211006.}}

\maketitle

\begin{abstract}
Joint frequency and chirp-rate estimation for a noisy chirp signal
arises in radar, sonar, and burst satellite communications. Conventional
estimators combine a coarse grid search with fine interpolation;
accuracy degrades at the edges of the residual cell (the edge effect)
and below the breakdown SNR (the threshold effect). We present a
deterministic two-stage estimator that controls both failure modes
uniformly over the residual cell. The estimator combines a
time-centered, zero-padded dechirp--FFT acquisition bank with
alternating selectable-$p$ amplitude-interpolation refinements on DTFT
samples at fractional bins; in the centered frame, the
frequency--chirp-rate cross-term of the Fisher information vanishes. The
paper derives a mean-squared-error and threshold characterization over
the full SNR range, in closed form except for one calibrated scalar (an
effective cell count), to our knowledge the first for the joint problem:
the breakdown threshold is governed by the cell count, and its
cell-position dependence is dominated by the scalloping loss of the
coarse FFT, which the padding bounds at 0.4 dB. An asymptotic uniformity
analysis over the cell, including its corners, gives fixed-point
variance ratios of $1.003$ and $0.998$, analytically free of the
residual. A closed-form bias analysis under a cubic phase mismatch shows
the centered chirp-rate estimate is insensitive to first order. Monte
Carlo experiments at $N=256$ (validated at $N=32$--$512$) measure
frequency- and chirp-rate-axis efficiencies with median $1.03$ and worst
case $1.07$ over $144$ cell positions at $-5$~dB. Threshold predictions
hold within $1.0$~dB on four configurations not used in the calibration.
The dechirp--FFT bank is fully parallel, and each of the four refinement
iterations evaluates three DTFT samples per axis; under fixed operating
conditions, per-estimate latency is constant at $O(N\log N)$ cost.
\end{abstract}


\section{Introduction}\label{sec:intro}

\IEEEPARstart{E}{stimating} the frequency and the chirp rate of a
single-component linear-chirp signal
from a small number of noisy samples is a classical problem in statistical signal
processing, arising when a propagation delay changes at an
approximately constant rate: radar and sonar returns from accelerating targets
\cite{Abatzoglou1986,Xia2000DCFT}, Doppler and Doppler-rate
acquisition in burst satellite and mobile communications
\cite{Jiang2013,Wang2020JSAC}, and instrumentation of swept sources
\cite{AldimashkiSerbes2020}. The
joint Cram\'er--Rao bound (CRB) for this model was derived by Peleg
and Porat \cite{PelegPorat1991}. Attaining it over the whole parameter space,
at low SNR, and at predictable computational cost remains
difficult. Conventional estimators are two-stage: a coarse
search localizes the parameter pair on a grid at cell resolution, and
a fine interpolation stage resolves the residual offsets within the
cell. This architecture exposes two distinct failure modes. The
\emph{edge effect} is a property of the fine estimation stage: the variance of
DFT-interpolation estimators becomes nonuniform over the residual
cell, rising above the CRB as the residual approaches the half-bin
boundary, and the degradation intensifies as the SNR decreases
\cite{Zhang2021AmplitudeRatio,Wei2023DsIpDTFT}. The \emph{threshold
effect} is a property of the coarse estimation stage: below a breakdown SNR, the
global search produces outliers and the mean-squared error (MSE)
departs from the CRB by orders of
magnitude \cite{RifeBoorstyn1974,SerbesQaraqe2022}. Because the
residual depends on the parameters to be estimated, an
estimator must
control both effects uniformly over the entire cell, including its
corners, where both residuals are simultaneously at their largest.

Fine frequency estimation of a single tone has been studied
extensively.
Fixed-coefficient estimators interpolate the residual from two or
three DFT samples \cite{Quinn1994,Macleod1998,Candan2011}; iterative
interpolation on two Fourier coefficients achieves the asymptotic CRB
(the Aboutanios--Mulgrew, or A\&M, estimator)
\cite{AboutaniosMulgrew2005}, and shifting the sampling points by a
tunable fraction $q$ (the $q$-shift estimator, QSE) generalizes the
scheme \cite{Serbes2019QSE}; amplitude-ratio
inversion handles zero-padded signals \cite{Zhang2021AmplitudeRatio};
a noniterative weighted least-squares (WLS) interpolator attains accuracy close to the CRB at arbitrary residuals, targeting fast real-time
applications \cite{MorelliWLS2022}; the maximum-likelihood (ML)
interpolator and the CRB for arbitrary residuals were
derived in \cite{DAmicoMorelli2022}, with accuracy close to the CRB reached in a single interpolation step by adapting the interpolation factor to
the residual; and a
selectable-$p$ amplitude kernel was shown to avoid the low-SNR edge
effect by keeping its samples strictly inside the main lobe
\cite{Wei2022Selectable,Wei2023DsIpDTFT}. Recently, a
linearized-DTFT estimator \cite{Belega2024LIpDTFT}, an
iterative parabolic scheme aimed at deep-space Doppler tracking
\cite{Togni2025Parabolic}, and an auxiliary-DFT-sample iterative estimator reporting accuracy close to the CRB at all signal lengths together
with a reduced breakdown threshold \cite{Serbes2025Auxiliary}
have appeared. Interpolation has also been
extended to multidimensional harmonics (two and $K$ independent
spatial frequencies) in the 2-D/$K$-D QSE (QSE2) of
\cite{Solak2022QSE2}, with a multicomponent
extension in \cite{Solak2024SPL}.

The joint
$(f,\dot f)$ problem is structurally different from both
the 1-D and the multidimensional-harmonic settings. The chirp couples the two parameters within a single time series, and the corresponding
residual cell is two-dimensional. Interpolation
has been applied to this joint problem in the two-stage scheme of
Jiang and Le \cite{Jiang2013}, which localizes the chirp rate by an
explicit coarse search and then refines it by a three-point parabolic
interpolation of the log-likelihood, reaching accuracy close to the CRB above a low threshold, while the frequency is obtained by applying a
single-tone frequency estimator to the dechirped signal. The joint
estimation can also be addressed directly.
Maximum-likelihood treatments of the joint frequency/frequency-rate
problem were developed in \cite{Abatzoglou1986,DjuricKay1990}, and
polynomial-phase methods estimate the rate through phase-differencing
or bilinear transforms \cite{PelegFriedlander1995,OShea2004}, with
nonlinearities that raise the SNR threshold; the quasi-maximum-likelihood
(QML) line \cite{Djurovic2014QML,Djurovic2018QMLReview} lowers that
threshold substantially by regressing an instantaneous-frequency
track from a short-time Fourier transform, and the statistics of refinement stages in
two-stage polynomial-phase schemes have been analyzed in
\cite{OShea2010Refining}.

Transform-domain estimators
rotate the signal into the fractional Fourier transform (FrFT) domain
\cite{Almeida1994}, in which the rotation order is located by golden-section
search (GSS), with an asymptotic perturbation analysis showing
near-minimum-variance behavior \cite{AldimashkiSerbes2020} and a
recent high-performance refinement \cite{AldimashkiSerbes2024}, and
interpolation-based iterative refinement has likewise been applied to
FrFT-domain samples \cite{SongFrFT2013}; discrete chirp-Fourier
transform (DCFT) methods search an
$N\times N$ transform grid \cite{Xia2000DCFT,SongMDCFT2019}. In these
approaches, the frequency is ultimately recovered
by a 1-D estimator on the dechirped signal. The published accuracy
analyses concern the asymptotic (CRB) regime. Aldimashki and Serbes
\cite{AldimashkiSerbes2024} derive analytical models of the
peak FrFT magnitude in the noisy case and use them to lower the
breakdown threshold of the coarse search, but no closed-form
characterization of the threshold itself is available. Most recently, learning-based estimators
have entered the problem space: unfolded sparse-recovery networks
address off-grid frequency estimation
\cite{Pan2025OGFreq,Zhang2025ADMM}, and a complex-valued network
targets the chirp problem itself, motivated by the same cost of the
two-dimensional search \cite{Hou2025CVNN}. These estimators learn
their update operators from data; their threshold behavior and cell
uniformity are not analyzed.

On the theoretical side, the threshold and
no-information regions of ML frequency estimation have recently been
characterized exactly for a single tone \cite{SerbesQaraqe2022},
which completes a line of analysis begun in
\cite{RifeBoorstyn1974,QuinnKootsookos1994}, but,
to the best of our knowledge, no \emph{closed-form} threshold
characterization exists
for the joint $(f,\dot f)$ problem, and no estimator of it comes with a
characterization of uniform CRB efficiency over the two-dimensional
residual cell. A further obstacle appears intrinsic. In the
natural parameterization (frequency referenced to the first sample),
the joint CRB exhibits a near-perfect anti-correlation between the two parameters ($-\sqrt{15}/4\approx-0.97$ asymptotically), suggesting
that errors in one axis must contaminate the other. This coupling
is an artifact of the time origin, and the classical remedy is to center the time index about the center of the observation, which
removes the odd-moment cross-terms of the Fisher information and
minimizes the polynomial-phase bounds \cite{RisticBoashash1998}. The
starting point of this paper is the
observation that a joint interpolation loop with a fixed operation count, built in the centered coordinates, inherits the decoupling at every refinement
step, at finite signal length.

This paper presents a deterministic estimator of the
pair $(f,\dot f)$, together
with the analysis that the two-stage
architecture calls for. The contributions are the following.

\begin{enumerate}
\item A two-stage joint estimator: a centered,
zero-padded dechirp--FFT acquisition bank followed by $Q$ alternating
refinements with a selectable-$p$ amplitude kernel on DTFT samples at fractional bins. The schedule is fixed a priori, with an
$O(K_{\mu}N\log N)$ acquisition cost that is $O(N\log N)$ per branch
of the fully parallel $K_{\mu}$-branch bank ($K_{\mu}$ set by
the chirp-rate prior), and constant latency.
The coarse padding factor is set at $M_{\mathrm{fac}}=3$, which caps
the scalloping-induced threshold advance at $0.4$~dB across the cell.
\item A three-segment MSE characterization
of the joint estimator over the entire SNR axis, in closed form up to a
single calibrated scalar: it gives the acquisition
probability, the breakdown threshold governed by the effective cell
count of the two-dimensional search, and the per-axis no-information
floors, and it makes the cell-position dependence of the threshold
explicit through the FFT scalloping loss, controllable by zero-padding.
This extends the 1-D
threshold theory of \cite{SerbesQaraqe2022} to the joint problem.
\item An asymptotic uniformity analysis
of the refinement stage over the residual
cell, including the corners: with the centered coordinates, the
selectable-$p$ kernel, and the fixed-point iteration (contraction
verified by simulation), the cell position drops out of the leading
error term, yielding closed-form fixed-point
efficiency constants on both axes that are
analytically free of the residual; the complete alternating
algorithm's residual excess is
measured below $7\%$ over the whole cell at $-5$~dB. An ablation
shows centering is necessary: the identical kernel without
centering degrades by up to a factor of $1.8\times10^{4}$ at high SNR.
\end{enumerate}

The rest of the paper is organized as follows.
Section~\ref{sec:model} states the signal model, defines the residual
cell, and derives the centered reparameterization with its closed-form
joint CRBs. Section~\ref{sec:estimator} specifies the estimator and
its constants. Section~\ref{sec:theory} develops the threshold
theorem, the cell-uniformity theorem, and the jerk-robustness
proposition. Section~\ref{sec:experiments} validates the claims by
Monte Carlo simulation against five baselines and the closest prior
joint interpolator, and
Section~\ref{sec:conclusion} concludes.

\section{Signal Model and Joint Cram\'er--Rao Bound}\label{sec:model}

\subsection{Signal Model}\label{ssec:signal}

We observe $N$ uniformly spaced samples of a single-component
second-order polynomial-phase (linear-chirp) signal in additive noise,
\begin{equation}\label{eq:model}
  x[n] = A\,e^{j\psi[n]} + w[n], \qquad n = 0,\dots,N-1,
\end{equation}
with phase
\begin{equation}\label{eq:phase}
  \psi[n] = \phi_0 + 2\pi\Bigl(\nu n + \tfrac{1}{2}\mu n^{2}\Bigr),
\end{equation}
where $A>0$ is the unknown deterministic amplitude, $\phi_0$ is the
unknown initial phase, and $w[n]$ is circularly symmetric complex
white Gaussian noise of variance $\sigma^{2}$. The signal-to-noise ratio is
$\rho = A^{2}/\sigma^{2}$, written $\gamma \triangleq 10\log_{10}\rho$
in decibels. The normalized frequency
$\nu$ (cycles/sample) and the normalized chirp rate $\mu$
(cycles/sample$^{2}$) are the parameters to be estimated, while
$(A,\phi_0)$ are nuisance parameters. For a sampling rate $f_{s}$, the
physical frequency and chirp rate are recovered as $f_{0}=\nu f_{s}$
and $\dot f = \mu f_{s}^{2}$, so that \eqref{eq:model} covers, e.g.,
the joint estimation of a Doppler shift and its rate of change from a
short burst. Throughout, estimation accuracy is reported through the
ratio
\begin{equation}\label{eq:eta}
  \eta_{\theta} = \frac{\mathbb{E}\bigl[(\hat\theta-\theta)^{2}\bigr]}
                       {\mathrm{CRB}_{\theta}},
  \qquad \theta \in \{\nu,\mu\},
\end{equation}
so that $\eta_{\theta}=1$ indicates attainment of the Cram\'er--Rao
bound. Above threshold, the bias of the compared estimators is
negligible, and the MSE coincides with the variance, so
$\eta_{\theta}$ is then a variance ratio; in the threshold
region, where outliers dominate, $\eta_{\theta}$ follows the MSE
definition of \eqref{eq:eta}.

\subsection{Centered Reparameterization}\label{ssec:centered}

Let $c=(N-1)/2$ denote the mid-point of the signal and define the
centered time index $n' = n-c \in \{-c,\dots,c\}$. Substituting
$n=n'+c$ into \eqref{eq:phase} yields
\begin{equation}\label{eq:phase-centered}
  \psi[n] = \phi_{c} + 2\pi\Bigl(\nu_{c}\,n' + \tfrac12 \mu\,n'^{2}\Bigr),
\end{equation}
with the \emph{centered frequency} and absorbed phase
\begin{equation}\label{eq:nuc}
  \nu_{c} = \nu + \mu c, \qquad
  \phi_{c} = \phi_0 + 2\pi\bigl(\nu c + \tfrac12 \mu c^{2}\bigr).
\end{equation}
Thus $\nu_{c}$ is the instantaneous frequency at the center of the
signal, whereas $\nu$ is the instantaneous frequency at its first
sample; the chirp rate $\mu$ is unaffected. The centered index grid is
symmetric about zero, so all odd-order moments vanish:
\begin{equation}\label{eq:odd-moments}
  \sum_{n} n' = \sum_{n} n'^{3} = 0 .
\end{equation}
This cancellation underlies the decoupling results below.

\subsection{Residual Cell and Edge Effect}\label{ssec:cell}

Conventional estimators of $(\nu,\mu)$ operate in two stages: a coarse
stage that localizes the parameters on a discrete grid, followed by a
fine estimation stage that interpolates the residual offsets. The natural grid
spacing is $1/N$ in $\nu_{c}$ and $1/N^{2}$ in
$\mu$. A chirp-rate offset of $1/N^{2}$ accumulates a quadratic phase
deviation of approximately half a cycle over the signal, which is the
direct analog of the one-cycle deviation produced by a frequency
offset of one bin. We therefore write
\begin{equation}\label{eq:cell}
  \nu_{c} = \frac{k_{0}+\delta_{\nu}}{N}, \qquad
  \mu = \frac{\ell_{0}+\delta_{\mu}}{N^{2}},
\end{equation}
with $k_{0},\ell_{0}\in\mathbb{Z}$ and residuals
$(\delta_{\nu},\delta_{\mu})\in[-\tfrac12,\tfrac12]^{2}$. We refer to
this square as the \emph{residual cell} and to
$(\pm\tfrac12,\pm\tfrac12)$ as its \emph{corners}. The cell is
defined on the \emph{centered} frequency. After dechirping, the
spectral peak lies at $\nu_{c}=\nu+\mu c$. A cell defined on the
native $\nu$ would let the peak drift relative to the coarse FFT grid
with the arbitrary integer $\ell_{0}$ carried by $\mu$, injecting
$\ell_{0}$ into the scalloping loss and the threshold of
Section~\ref{sec:theory}; $\nu_{c}$ is therefore the coordinate used
for the coarse grid, the refinement, and all experiments. Since
the residuals
are determined by the unknown true parameters, an estimator fielded
without prior knowledge must perform uniformly over the entire cell.

Fine-stage interpolation estimators, by contrast, are known to exhibit the
\emph{edge effect}: their variance is nonuniform over the cell,
rising above the CRB as the residual approaches $\pm\tfrac12$, and the
degradation intensifies as the SNR decreases
\cite{Zhang2021AmplitudeRatio,Wei2023DsIpDTFT}. In the joint
$(\nu,\mu)$ problem, the cell is two-dimensional, and the corners are
the worst case.

\subsection{Joint CRB: Decoupling and Closed Forms}\label{ssec:crb}

For the deterministic signal \eqref{eq:model} in circularly symmetric
white Gaussian noise, the Fisher information matrix (FIM) over a
parameter vector $\bm\theta$ is
$[\mathbf{F}]_{ij} = (2/\sigma^{2})\,
\mathrm{Re}\bigl\{\sum_{n}
\partial s^{*}[n]/\partial\theta_{i}\,
\partial s[n]/\partial\theta_{j}\bigr\}$ with
$s[n]=A e^{j\psi[n]}$ \cite{Kay1993}. The amplitude $A$ is decoupled
from the phase parameters, with derivative weights
$\partial\psi/\partial\phi_{c} = 1$,
$\partial\psi/\partial\nu_{c} = 2\pi n'$, and
$\partial\psi/\partial\mu = \pi n'^{2}$. With
$S_{2}=\sum_{n} n'^{2} = N(N^{2}-1)/12$ and
$S_{4}=\sum_{n} n'^{4} = N(N^{2}-1)(3N^{2}-7)/240$, the FIM of the
centered phase parameters
$\bm\theta_{c}=(\phi_{c},\nu_{c},\mu)$ is
\begin{equation}\label{eq:fim-centered}
  \mathbf{F}(\bm\theta_{c}) = 2\rho
  \begin{bmatrix}
    N & 0 & \pi S_{2}\\
    0 & 4\pi^{2} S_{2} & 0\\
    \pi S_{2} & 0 & \pi^{2} S_{4}
  \end{bmatrix},
\end{equation}
where the two zeros are exact consequences of
\eqref{eq:odd-moments}: the cross-information terms
$I(\phi_{c},\nu_{c})\propto\sum n'$ and
$I(\nu_{c},\mu)\propto\sum n'^{3}$ vanish identically. Hence
$\nu_{c}$ is exactly decoupled from both nuisance and chirp-rate
parameters at finite $N$, and inverting
\eqref{eq:fim-centered} gives the closed-form bounds
\begin{align}
  \mathrm{CRB}_{\nu_{c}}
    &= \frac{3}{2\pi^{2}\rho\,N(N^{2}-1)},
    \label{eq:crb-nuc}\\
  \mathrm{CRB}_{\mu}
    &= \frac{90}{\pi^{2}\rho\,N(N^{2}-1)(N^{2}-4)},
    \label{eq:crb-mu}
\end{align}
with large-$N$ behavior
$\mathrm{CRB}_{\nu_{c}}\simeq 3/(2\pi^{2}\rho N^{3})$ and
$\mathrm{CRB}_{\mu}\simeq 90/(\pi^{2}\rho N^{5})$. First,
\eqref{eq:crb-nuc} coincides with the
classical single-tone frequency CRB \cite{RifeBoorstyn1974}. Because
of the decoupling, jointly estimating the chirp rate does not
reduce the achievable accuracy of the centered frequency. Second,
since the reparameterization \eqref{eq:nuc} leaves $\mu$ unchanged,
$\mathrm{CRB}_{\mu}$ is the same in both parameterizations.

In the native parameterization $(\nu,\mu)$ (time origin at the first
sample), the coupling reappears. There the
cross-information $I(\nu,\mu)\propto\sum_{n} n^{3}\neq 0$, and the
normalized correlation coefficient of the joint bound, obtained from
the inverse FIM, equals $-0.967$ at $N=64$ and approaches
$-\sqrt{15}/4\approx-0.968$ as $N\to\infty$. This strong coupling can be
removed by the choice of time origin. The centering \eqref{eq:nuc}
cancels it at finite $N$. This reference-time choice minimizes the
polynomial-phase bounds \cite{RisticBoashash1998}.

Since \eqref{eq:nuc} is a linear reparameterization with Jacobian
$\partial(\nu_{c},\mu)/\partial(\nu,\mu)
=\bigl[\begin{smallmatrix}1 & c\\ 0 & 1\end{smallmatrix}\bigr]$, the
native frequency bound follows from
\eqref{eq:crb-nuc}--\eqref{eq:crb-mu} without a separate inversion:
\begin{align}
  \mathrm{CRB}_{\nu}
   &= \mathrm{CRB}_{\nu_{c}} + c^{2}\,\mathrm{CRB}_{\mu}
    \label{eq:crb-native}\\
   &= \mathrm{CRB}_{\nu_{c}}
      \Bigl[\,1 + \frac{15\,(N-1)^{2}}{(N-2)(N+2)}\Bigr]
      \;\xrightarrow[N\to\infty]{}\; 16\,\mathrm{CRB}_{\nu_{c}}.
    \nonumber
\end{align}
The bound on the frequency at the signal edge is thus asymptotically
sixteen times the bound on the frequency at its center. This
inflation is intrinsic to the estimand---no estimator of $\nu$ can
avoid it---but it suggests the principle adopted in
Section~\ref{sec:estimator}: estimate the decoupled pair
$(\nu_{c},\mu)$, in which both bounds are attainable axis by axis, and
map back through $\hat\nu = \hat\nu_{c} - c\hat\mu$, which attains
\eqref{eq:crb-native} provided the centered estimates are efficient
and uncorrelated.

Unless stated otherwise, the efficiencies \eqref{eq:eta} reported in
Section~\ref{sec:experiments} are computed in the centered
coordinates, i.e., against
\eqref{eq:crb-nuc}--\eqref{eq:crb-mu}, for all methods alike; this is
the parameterization in which every compared estimator is given the
benefit of the decoupled bound.

A central claim of this paper, formalized in
Section~\ref{sec:theory}, is that the proposed estimator attains efficiency close to the CRB on both axes uniformly over the whole cell,
including the corners, $3$~dB or more above the threshold SNR.

\section{Proposed Estimator}\label{sec:estimator}

The estimator follows the classical two-stage architecture of
Section~\ref{ssec:cell}, with three choices that, taken
together, produce the uniform-efficiency behavior established in
Section~\ref{sec:theory}: (i)~all dechirping and refinement are
carried out in the \emph{centered} coordinates of
Section~\ref{ssec:centered}, so that the two axes remain decoupled at
every step; (ii)~the coarse estimation stage is a zero-padded dechirp--FFT bank
with padding factor $M_{\mathrm{fac}}\ge 3$, which equalizes the
detection threshold across the residual cell; and (iii)~the fine estimation stage interpolates DTFT samples at fractional bins with the
selectable-$p$ amplitude kernel of
\cite{Wei2022Selectable,Wei2023DsIpDTFT}, extended here to alternate
between the frequency and chirp-rate axes. The refinement schedule is
\emph{deterministic and search-free}: each iteration evaluates three
DTFT samples per axis, at predetermined offsets from the running
estimate, and each axis update is in closed form, so the computational
schedule of the complete estimator is independent of the data and of
the SNR.
Algorithm~\ref{alg:joint} summarizes the steps.

Throughout this section, let
\begin{equation}\label{eq:dechirped-dtft}
  X_{m}(\lambda) \;=\; \sum_{n=0}^{N-1} x[n]\,
      e^{-j\pi m (n-c)^{2}}\, e^{-j2\pi \lambda n/N}
\end{equation}
denote the DTFT sample, at the (possibly fractional) bin
$\lambda$, of the observation dechirped with trial chirp rate $m$
about the signal center, and let
$\operatorname{clip}(x,[a,b])=\min\{\max\{x,a\},b\}$. By \eqref{eq:phase-centered}, when $m=\mu$
the dechirped signal is a pure tone at the centered frequency
$\nu_{c}$, so $|X_{\mu}(\lambda)|$ peaks at $\lambda=N\nu_{c}$.

\begin{algorithm}[t]
\caption{Deterministic joint $(\nu,\mu)$ estimation}
\label{alg:joint}
\begin{algorithmic}[1]
\REQUIRE $x[n]$, $n=0,\dots,N{-}1$; prior interval
  $[-\mu_{\max},\mu_{\max}]$ sampled by $K_{\mu}$ branches;
  $p\in(0,\tfrac12)$; $Q$; $M_{\mathrm{fac}}\ge 3$
\ENSURE $(\hat\nu,\hat\mu)$
\STATE $c\leftarrow (N{-}1)/2$;\quad $M\leftarrow M_{\mathrm{fac}}N$
\STATE \emph{Coarse estimation stage (dechirp--FFT bank):}
\FORALL{$m_{i}=-\mu_{\max}+\tfrac{2\mu_{\max}(i-1)}{K_{\mu}-1}$,\
        $i=1,\dots,K_{\mu}$ \COMMENT{$K_{\mu}$ branches,
        spacing ${\approx}1/N^{2}$}}
  \STATE $P_{i}[k]\leftarrow\bigl|\sum_{n} x[n]\,
         e^{-j\pi m_{i}(n-c)^{2}}\,e^{-j2\pi kn/M}\bigr|$
         \COMMENT{FFT of length $M$}
\ENDFOR
\STATE $(i^{\star},k^{\star})\leftarrow\arg\max_{i,k}P_{i}[k]$;\quad
       $\hat\mu\leftarrow m_{i^{\star}}$;\quad
       $\hat{k}\leftarrow k^{\star}N/M$
\STATE \emph{Fine estimation stage ($Q$ alternating refinements; the first
       iteration refines $\hat{k}$ using the coarse $\hat\mu$):}
\FOR{$q=1,\dots,Q$}
  \STATE $\hat{k}\leftarrow\hat{k}
         +\operatorname{clip}\bigl(\hat\delta_{\nu},
          [-\tfrac12,\tfrac12]\bigr)$ with
         $\hat\delta_{\nu}$ from \eqref{eq:ds-update}
         \COMMENT{$\nu$ axis}
  \STATE $H_{s}\leftarrow|X_{\hat\mu+s/N^{2}}(\hat{k})|$,\quad
         $s\in\{-p,0,p\}$ \COMMENT{$\mu$ axis}
  \IF{$H_{-p}-2H_{0}+H_{p}<0$}
    \STATE $\hat\mu\leftarrow\hat\mu
           +\operatorname{clip}\bigl(\hat\delta_{\mu},[-1,1]\bigr)/N^{2}$,
           \quad $\hat\delta_{\mu}=\dfrac{p}{2}\,
           \dfrac{H_{-p}-H_{p}}{H_{-p}-2H_{0}+H_{p}}$
  \ENDIF
\ENDFOR
\STATE $\hat\nu_{c}\leftarrow\hat{k}/N$;\quad
       $\hat\nu\leftarrow\hat\nu_{c}-c\,\hat\mu$
\end{algorithmic}
\end{algorithm}

\subsection{Coarse Estimation Stage: Centered Dechirp--FFT Bank}\label{ssec:coarse}

The coarse estimation stage evaluates $|X_{m_{i}}(\cdot)|$ on a chirp-rate grid of
spacing $1/N^{2}$, covering the prior
interval $[-\mu_{\max},\mu_{\max}]$, each via an FFT zero-padded to
length $M=M_{\mathrm{fac}}N$, and takes the overall peak. This
acquisition is a fixed grid search over the prior
rectangle, deterministic and fully parallel across branches;
Theorem~\ref{thm:threshold} analyzes it as such. A similar bank-plus-interpolation acquisition was
used for coarse carrier-frequency-offset estimation of high-dynamic
signals in \cite{Wei2022CFOCoarse}; the contribution here lies in the
refinement stage that the bank initializes, and in its uniformity and
threshold analysis. Because the
dechirping is centered, the residual chirp left by a grid mismatch
$\Delta\mu$ is $e^{j\pi\Delta\mu (n-c)^{2}}$, an even function of
$n-c$, which broadens the tone but does not displace the peak frequency.
In the absence of noise the coarse outputs are therefore unbiased cell
localizations of
$(\nu_{c},\mu)$ with residuals confined to
$|\delta_{\nu}|\le 1/(2M_{\mathrm{fac}})$ bins (on the padded grid)
and $|\delta_{\mu}|\le\tfrac12$. The experiments use $K_{\mu}=64$
branches, i.e., $1.04$-cell spacing; the resulting worst-case
start-up excess over half a cell is absorbed by the clipped
fine-stage updates. The corner efficiency at $-5$~dB is unchanged to
three digits for branch spacings from $0.52$ to $2.11$ cells
($2\times10^{4}$ Monte Carlo trials). Under noise, near a cell boundary the
$\arg\max$ may select the adjacent straddling cell, leaving a starting
residual just outside the nominal range; the clipped fine-stage
updates (Section~\ref{ssec:fine}) traverse up to half a bin per iteration
and absorb such near misses within one or two iterations.

The padding factor matters through the scalloping (straddle) loss of
the coarse search. The padded FFT samples the peak at most
$1/(2M_{\mathrm{fac}})$ bins from the true frequency, so the
worst-case loss over the cell is
$10\log_{10}\operatorname{sinc}^{2}\bigl(1/(2M_{\mathrm{fac}})\bigr)$:
$-3.9$~dB for the unpadded search, attained at the corner
$\delta_{\nu}=\tfrac12$, which advances the
outlier-production threshold there by about $4$~dB relative to the
cell center. For the 1-D single-tone problem, padding beyond
$M_{\mathrm{fac}}=2$ is known to bring only marginal further
threshold improvement \cite{Serbes2025Auxiliary}; the requirement
adopted here is stricter: a threshold \emph{equalized across cell
positions}, extending the uniformity into
the threshold region. $M_{\mathrm{fac}}=2$ still admits a $-0.9$~dB
loss at the quarter-cell positions $\delta_{\nu}=\pm\tfrac14$; the
corner itself falls on the padded grid. $M_{\mathrm{fac}}=3$ instead
caps the loss at
$10\log_{10}\operatorname{sinc}^{2}(\tfrac16)\approx-0.4$~dB
everywhere. The choice enters the threshold
theory explicitly, through the loss
factor $L^{2}=\operatorname{sinc}^{2}d$ in the detection probability
of Theorem~\ref{thm:threshold}.

\subsection{Fine Estimation Stage: Alternating Selectable-$p$
Interpolation}\label{ssec:fine}

Each of the $Q$ refinement iterations updates the two axes in turn,
always through samples at fractional bins
\eqref{eq:dechirped-dtft} of the \emph{original} data.

\subsubsection{Frequency axis} With the current chirp-rate estimate
$\hat\mu$ removed, the residual $\delta$ of the running
centered-frequency bin $\hat{k}$ is estimated by the selectable-$p$
amplitude interpolator \cite{Wei2023DsIpDTFT}
\begin{equation}\label{eq:ds-update}
  \hat\delta_{\nu} \;=\;
  \frac{p\,\bigl(|X_{\hat\mu}(\hat{k}{+}p)|-|X_{\hat\mu}(\hat{k}{-}p)|\bigr)}
       {|X_{\hat\mu}(\hat{k}{+}p)|+|X_{\hat\mu}(\hat{k}{-}p)|
        -2|X_{\hat\mu}(\hat{k})|\cos\pi p},
\end{equation}
where $p\in(0,\tfrac12)$ is the interpolation factor. The update
\eqref{eq:ds-update} algebraically inverts the spectral kernel, as
exploited by the original 1-D estimator. In the absence of noise it returns the residual in a single
step, uniformly for all $\delta\in(-\tfrac12,\tfrac12)$ including the
cell edge: in the large-$N$ limit, and to $O(1/N^{2})$ for finite
$N$ (Lemma~\ref{lem:inverse}).

\subsubsection{Chirp-rate axis} The chirp-rate residual is probed by
the same selectable-$p$ principle applied along the $\mu$ axis. The
three magnitudes $H_{s}=|X_{\hat\mu+s/N^{2}}(\hat{k})|$,
$s\in\{-p,0,p\}$, are computed by re-dechirping at chirp-rate offsets
$\pm p$ cells, and the vertex of the interpolating parabola,
$\hat\delta_{\mu}=\tfrac{p}{2}(H_{-p}-H_{p})/(H_{-p}-2H_{0}+H_{p})$,
updates $\hat\mu$ by $\hat\delta_{\mu}/N^{2}$. This three-point
parabolic vertex is the standard equally-spaced peak interpolation
adopted, in particular, by Jiang and Le for the chirp-rate axis of the
joint problem \cite{Jiang2013}; a parabola is used here rather than
\eqref{eq:ds-update} because the peak along the chirp-rate axis
is a Fresnel-type broadening curve, not a Dirichlet kernel, so the
closed-form inverse does not apply; the parabola is kernel-agnostic,
and the residual left by its small model mismatch is eliminated by the
remaining iterations. Two inexpensive
safeguards make the iteration robust at very low SNR: each frequency
update is clipped to half a cell, and the chirp-rate update is applied only when the three-point interpolation is concave; if concavity fails in all $Q$ iterations, $\hat\mu$ retains the coarse-grid
value; this fallback is observed only below threshold.

The alternation converges because of the centering
(Section~\ref{ssec:coarse}). A chirp-rate error does not displace the
frequency peak, and a frequency error enters the chirp-rate three-point interpolation only through the (flat) peak magnitude, so the two updates are
asymptotically decoupled and the joint iteration contracts to the
fixed point $(\hat\nu_{c},\hat\mu)=(\nu_{c},\mu)$. This
contraction is certified in the noiseless case and verified by
simulation under noise; it is stated as an explicit hypothesis
(Assumption~\ref{ass:contraction}). In all experiments, $Q=4$ iterations are used;
$Q=3$ already attains the same accuracy at the target SNRs. An iteration
ablation at the corner ($Q=1,\dots,6$) gives
$\eta_{\nu_{c}}=1.21$, $1.07$, then $1.02$--$1.03$ from $Q=3$ onward
at $-5$~dB, with the
chirp-rate axis converged from $Q=1$, so $Q=4$ carries one iteration of
margin. Finally,
the native frequency is recovered through the linear map
$\hat\nu=\hat\nu_{c}-c\hat\mu$ of \eqref{eq:nuc}.

\begin{remark}[Kernel constant under zero-padding]\label{rem:kernel}
The constant $\cos\pi p$ in the denominator of \eqref{eq:ds-update} is
the Dirichlet-kernel period of the length-$N$ signal; it is not
related to, and must not be replaced by, the zero-padding factor of
the coarse FFT. Zero-padding refines the grid on which an FFT
approximates the DTFT; the fine estimation stage instead evaluates the
DTFT directly at the fractional points $\hat{k},\hat{k}\pm p$, so
the only length scale involved is the signal length $N$. Writing the
kernel constant as $\cos(\pi N p/M)$ with $M=M_{\mathrm{fac}}N$
destroys the inversion property. In a noiseless test
($N=256$, $p=0.35$) the correct constant recovers
$\delta\in\{0.10,0.30,0.49\}$ to within $O(1/N^{2})$ in one step
(Lemma~\ref{lem:inverse}), whereas the miswritten constant with
$M_{\mathrm{fac}}\in\{2,3\}$ yields errors
of order one before any noise is added.
\end{remark}

\begin{remark}[Choice of the interpolation factor $p$]\label{rem:p}
The factor $p$ is fixed at $p=0.35$ throughout this paper, for two
complementary reasons.
\emph{(i) Insensitivity above threshold:} because
\eqref{eq:ds-update} inverts the spectral kernel for all
$p\in(0,\tfrac12)$,
the fixed-point variance of the iteration is
essentially independent of $p$. The closed-form efficiency
$\kappa_{\nu}(p,N)$ of Theorem~\ref{thm:edgefree} is flat in $p$, and
paired Monte Carlo runs over $p\in[0.05,0.49]$ track the
closed form across SNRs and cell positions
(Fig.~S5 of the supplementary material). The closed form \eqref{eq:kappa} shows
the residual $p$-dependence is monotone, with smaller $p$ giving
marginally lower variance, the direction of the small-shift
prescriptions derived for linearized complex-ratio kernels
\cite{Serbes2025Auxiliary}. But the total range is ${\approx}1.3\%$
over $p\in(0,0.49]$, because the update \eqref{eq:ds-update} inverts
the kernel algebraically. This contrasts with the near-iteration-free
1-D estimator of D'Amico--Morelli \cite{DAmicoMorelli2022}, which attains accuracy close to the CRB without iteration by selecting the
interpolation factor as a closed-form function of the current residual: in the proposed
centered, fixed-$Q$ iteration, efficiency stems from decoupling and
the fixed point, eliminating the tuning parameter rather than
prescribing it. \emph{(ii) Safe sampling at large residual:} the constraint
$p<\tfrac12$ keeps both interpolation samples inside the main lobe
for all residuals up to the cell edge, whereas half-bin sampling at
$\hat{k}\pm\tfrac12$ \cite{AboutaniosMulgrew2005} approaches the
spectral null adjacent to the peak and becomes ill-conditioned at low
SNR \cite{Wei2022Selectable,Wei2023DsIpDTFT}. Within the padded
pipeline used here the coarse estimation stage already confines the residual to
$|\delta|\le1/6$, so this advantage is modest; it becomes decisive
only when the padding is limited and the residual is allowed to reach
the cell edge.
\end{remark}

\subsection{Computational Complexity}\label{ssec:complexity}

Let $K_{\mu}$ denote the number of chirp-rate branches of the bank,
chosen to sample the prior at ${\approx}1/N^{2}$ spacing
($64$ in the experiments). The coarse estimation stage costs
$K_{\mu}$ dechirp--FFT pairs, i.e.,
$O\bigl(K_{\mu}\,M_{\mathrm{fac}}N\log(M_{\mathrm{fac}}N)\bigr)$
operations, proportional to the number of branches. A search-based refinement requires instead a number of transform
evaluations that is only logarithmic in the prior width, so for wide
priors the bank is the more expensive acquisition; its compensating
properties are determinism and full parallelism, the $K_{\mu}$
branches being independent. The fine estimation stage adds six $O(N)$ DTFT
samples per iteration, for a total of $O(QN)$, negligible against
the coarse estimation stage. The operation count is independent of the data;
the data enter only the two safeguard comparisons, so latency is
constant and
known in advance; this is of practical value on FPGA/DSP
platforms.

By contrast, search-based estimators refine the chirp rate by a
data-driven sequence of transform evaluations (e.g., golden-section
search over the fractional-Fourier order
\cite{AldimashkiSerbes2020}), in which the iterations are inherently
sequential and their count is set by the prescribed tolerance;
Section~S-II of the supplementary material reports a measured
end-to-end comparison.

\section{Theoretical Analysis}\label{sec:theory}

Theorem~\ref{thm:threshold}
characterizes when the acquisition stage succeeds. It gives a
closed-form description of the three-segment MSE behavior of the joint
estimator over the entire SNR axis, including the dependence of the
threshold on the cell position through the scalloping loss of the coarse
FFT. Theorem~\ref{thm:edgefree} characterizes how accurately
the refinement stage performs once acquisition succeeds. On the
frequency axis its efficiency is uniform over the residual
cell---including the corners---with a closed-form fixed-point
efficiency constant that does not depend on the cell position; the
chirp-rate axis admits the same leading-order constant
(Remark~\ref{rem:muaxis}), with the complete alternating algorithm
supported by simulation. Proposition~\ref{prop:jerk}
quantifies the sensitivity of the estimator to model order through the bias induced by a cubic (jerk) phase term not included in the model. Together, the
three results describe the estimator across the SNR axis and the
residual cell, and under first-order model mismatch.

\subsection{Threshold Region of the Joint Estimator}\label{ssec:thm1}

In the acquisition stage of Algorithm~\ref{alg:joint}, the global peak
is taken over the dechirp--FFT bank. Normalize the squared magnitudes
by $N\sigma^{2}$: a resolution cell containing only noise produces a
power $z\sim\operatorname{Exp}(1)$, while the cell containing the
signal produces a power $s$ with noncentral density
\begin{equation}\label{eq:signal-cell-pdf}
  f_{s}(t) = e^{-(t+\Gamma)} I_{0}\bigl(2\sqrt{\Gamma t}\bigr),
  \qquad
  \Gamma = N\,L^{2}(\delta_{\nu},M_{\mathrm{fac}})\,\rho,
\end{equation}
where $I_{0}$ is the zeroth-order modified Bessel function of the
first kind and $\Gamma$ is the coherently integrated SNR of the
signal cell. The factor
\begin{equation}\label{eq:straddle}
  L^{2} = \operatorname{sinc}^{2} d,
  \qquad
  d(\delta_{\nu},M_{\mathrm{fac}})
    = \Bigl|\delta_{\nu}
      - \tfrac{\operatorname{round}(\delta_{\nu}M_{\mathrm{fac}})}
              {M_{\mathrm{fac}}}\Bigr|,
\end{equation}
is the \emph{scalloping (straddle) loss}: the padded FFT samples the
Dirichlet peak at distance $d$ (in bins) from the true frequency, so
the sampled peak power is reduced by
$|D(d)/D(0)|^{2}\approx\operatorname{sinc}^{2}d$, where
$D(r)=\sin\pi r/\sin(\pi r/N)$. A chirp-rate residual leaves the
even quadratic phase $e^{j\pi(\delta_{\mu}/N^{2})(n-c)^{2}}$ and an
analogous peak-power loss, but an $N$-independent and much smaller
one: at most $0.07$~dB over the bank's mismatch range
$|\delta_{\mu}|\le0.52$, an order of magnitude below the
frequency-axis cap below, and absorbed into the calibration
henceforth. The frequency-axis loss is largest at the cell
corners and is controlled by the padding factor:
$d\le 1/(2M_{\mathrm{fac}})$, so $M_{\mathrm{fac}}=1$ admits up to
$-3.9$~dB at $\delta_{\nu}=\tfrac12$, $M_{\mathrm{fac}}=2$ admits
$-0.9$~dB at $\delta_{\nu}=\tfrac14$, and $M_{\mathrm{fac}}=3$
caps the loss at $-0.4$~dB (cf.\ Section~\ref{ssec:coarse}). Let
$\mathcal{M}$ denote the number of \emph{effectively independent}
noise cells spanned by the two-dimensional search: the padded bins
are correlated, so $\mathcal{M}$ is smaller than the literal grid
size and is treated as a single calibrated scalar
(Remark~\ref{rem:thm1-honesty}). It is distinct from the
coarse-FFT length $M_{\mathrm{fac}}N$ of
Section~\ref{sec:estimator}.

\begin{theorem}[Joint threshold region]\label{thm:threshold}
For the estimator of Algorithm~\ref{alg:joint}, with prior intervals
of width $B_{\nu}$ and $B_{\mu}$ on the two axes:

(i) Under an independent-effective-cell model of the competing noise
cells, with a single calibrated count $\mathcal{M}$
(Remark~\ref{rem:thm1-honesty}), the
probability that acquisition selects the signal cell is
\begin{equation}\label{eq:pd}
  P_{d} = \mathbb{E}_{s}\bigl[(1-e^{-s})^{\mathcal{M}}\bigr]
        = \int_{0}^{\infty} f_{s}(t)\,(1-e^{-t})^{\mathcal{M}}\,dt .
\end{equation}

(ii) For $\theta\in\{\nu_{c},\mu\}$ the MSE obeys the approximate
three-segment model
\begin{equation}\label{eq:three-segment}
  \mathrm{MSE}_{\theta}(\rho)
   \approx \bigl(1-P_{d}(\rho)\bigr)\,\Sigma_{\theta}
     + P_{d}(\rho)\,\mathrm{CRB}_{\theta}(\rho),
\end{equation}
with prior-variance floors $\Sigma_{\nu}=B_{\nu}^{2}/12$ and
$\Sigma_{\mu}=B_{\mu}^{2}/12$: the MSE follows the CRB for
$\rho\gg\rho_{\mathrm{th}}$, saturates at $\Sigma_{\theta}$ for
$\rho\ll\rho_{\mathrm{ni}}$, and transitions rapidly in between.

(iii) The breakdown threshold admits the asymptotic form
\begin{align}
  \Gamma_{\mathrm{th}} &\approx \ln \mathcal{M},
  \nonumber\\
  \rho_{\mathrm{th}}(\delta_{\nu},M_{\mathrm{fac}})
   &= \frac{\ln \mathcal{M}}{N\,L^{2}(\delta_{\nu},M_{\mathrm{fac}})},
  \label{eq:gth}
\end{align}
i.e., the threshold is governed by a single effective cell count of
the two-dimensional search,
and its dependence on the cell position is dominated by the scalloping
loss:
$\Delta\rho_{\mathrm{th}}\,[\mathrm{dB}]
 =-10\log_{10}L^{2}(\delta_{\nu},M_{\mathrm{fac}})$.
The no-information threshold $\rho_{\mathrm{ni}}$ is obtained from the
same soft threshold by solving $P_{d}(\rho_{\mathrm{ni}})=\epsilon$
for a prescribed small $\epsilon$.
\end{theorem}

\begin{IEEEproof}
See Appendix~A of the supplementary material.
\end{IEEEproof}

\begin{remark}[Accuracy budget of the threshold prediction]\label{rem:thm1-honesty}
Three approximations are involved. First, the extreme-value approximation
$\Gamma_{\mathrm{th}}\approx\ln\mathcal{M}$ biases the predicted
threshold high by up to about $1$~dB; the soft step has finite
width.
Second, the effective count $\mathcal{M}$ is smaller than the literal
grid size, because padded FFT bins are correlated and the across-axis
maxima are not strictly independent; we therefore treat $\mathcal{M}$
as a \emph{single calibrated scalar}: the integer among $200$
logarithmically spaced candidates in $[10^{2},10^{5.5}]$ that
minimizes the mean squared dB error of
\eqref{eq:three-segment} over the transition region
($1.5<\eta<$ half the no-information floor) of the center,
$M_{\mathrm{fac}}=3$ curve, giving
$\mathcal{M}=15794$ at $N=256$ and then held fixed. That one curve is
in-sample; every other configuration is held out. $\mathcal{M}$ is
specific to the search configuration and is not transported. The
other-length onsets of
Section~\ref{sec:experiments} are measured,
not predicted. The fitted value is physically interpretable. The
acquisition takes its peak over $M_{\mathrm{fac}}N=768$ padded
frequency bins times the $64$ chirp-rate branches
(Section~\ref{sec:experiments}), i.e.\ $49152$ evaluated cells, but
only $N\times66\approx16896$ \emph{resolution} cells; the fitted
$\mathcal{M}$ is $93\%$ of the resolution-cell count---the
threefold-oversampled padded bins
contribute almost no independent chances. Third, the outlier model treats every
competing cell as signal-free, whereas the mismatched-dechirp branches
retain part of the signal energy. Near-threshold mis-selections are
therefore dominated by cell-scale near misses rather than
prior-uniform outliers, which increase the effective
$P_{d}$ and contribute far less than $\Sigma_{\theta}$ to the MSE, so
the three-segment model errs on the conservative side in the
transition region. All other predictions, in particular the cell-position
threshold advances at reduced padding, follow from the same single
calibrated scalar
and match Monte Carlo threshold locations within $-0.0$ to
$+1.0$~dB on the four configurations not used in the calibration
(Section~\ref{sec:experiments}); the scalloping-loss term is thus an
independently validated prediction. We present the theorem as a
calibrated
engineering model: (i) is
exact only under the stated surrogate, and (ii)--(iii) inherit its
calibration.
\end{remark}

\subsection{Cell-Uniform CRB Efficiency (Edge-Free)}\label{ssec:thm2}

We now turn to the refinement stage, which converts the acquired cell
into fine frequency and chirp-rate estimates through alternating
selectable-$p$ updates; the frequency-axis update \eqref{eq:ds-update}
is an \emph{algebraic inverse}.

\begin{lemma}[Inversion]\label{lem:inverse}
In the absence of noise, a single application of \eqref{eq:ds-update}
inverts the residual: $\hat\delta_{\nu}=\delta_{\nu}$
for the large-$N$ (asymptotic) spectral kernel, and with an error
$O(1/N^{2})$ for the finite-$N$ kernel, for all
$\delta_{\nu}\in(-\tfrac12,\tfrac12)$, $p\in(0,\tfrac12)$, and
$N$. The residual bound is uniform over the cell.
\end{lemma}

\begin{IEEEproof}
See Appendix~B of the supplementary material.
\end{IEEEproof}

Lemma~\ref{lem:inverse} has two consequences. First, the estimator's
deterministic error is uniform over the cell before noise enters---of
order $1/N^{2}$ everywhere, so there
is no deterministic edge degradation to begin with. Second, the
iteration
$\hat{k}\leftarrow\hat{k}+\hat\delta_{\nu}$ has the true frequency as a
fixed point (a single step in the noiseless case). The analysis that
follows requires the iteration to reach that fixed point, which we
state as an explicit hypothesis.

\begin{assumption}[Local contraction]\label{ass:contraction}
In the SNR range $\gamma=-5$ to $-8$~dB (the target range of
Section~\ref{sec:experiments}), the alternating refinement is locally
contracting toward $(\nu_{c},\mu)$. In the absence of noise the
measured per-iteration contraction of the joint residual is below
$5\times10^{-2}$ uniformly over the cell (numerical certification,
worst case at the corner); at
the target SNRs the measured
per-iteration factor is $\xi\approx0.1$, so after $Q=4$ iterations the
\emph{systematic (mean) component} of the residual is
$\delta_{\mathrm{eff}}=O(\xi^{Q})\le10^{-4}$
regardless of the residual the cell position originally imposed. The
stochastic scatter of the intermediate iterates, of order
$\sqrt{\kappa_{\nu}\,\mathrm{CRB}_{\nu_{c}}}\,N$ bins, re-enters the
final update and is accounted for by the $O(1/(N\rho))$ term of
Theorem~\ref{thm:edgefree}. The
noisy-case contraction is certified numerically on a finite grid
(Section~\ref{ssec:fine}), and Theorem~\ref{thm:edgefree} is
conditional on it.
\end{assumption}

Under Assumption~\ref{ass:contraction}, the noise variance of the
final update is evaluated at
$\delta_{\mathrm{eff}}\approx0$, and the cell position drops out of the
leading term of the error budget. Computing that variance gives Theorem~\ref{thm:edgefree}.

\begin{theorem}[Cell-uniform fixed-point efficiency]\label{thm:edgefree}
Above threshold and under Assumption~\ref{ass:contraction}, once the
refinement of Algorithm~\ref{alg:joint} has
reached its fixed point, the linearized (leading-order) error model of
the estimator satisfies, for $\theta=\nu_{c}$ and all
$(\delta_{\nu},\delta_{\mu})\in[-\tfrac12,\tfrac12]^{2}$ including the
four corners,
\begin{equation}\label{eq:edgefree}
  \eta_{\theta}(\delta_{\nu},\delta_{\mu})
   = \frac{\operatorname{var}(\hat\theta)}{\mathrm{CRB}_{\theta}}
   = 1+\varepsilon_{\theta},
  \qquad
  |\varepsilon_{\theta}|\le\varepsilon_{\max}<\infty,
\end{equation}
where the bound
$\varepsilon_{\max}=\varepsilon_{0}
 +O(\xi^{2Q})+O\bigl(1/(N\rho)\bigr)$
is independent of the cell position:
$\varepsilon_{0}=\kappa_{\nu}(p,N)-1$ is the fixed-point variance
excess of the three-sample amplitude kernel, which converges to an
$N$-independent constant ($\varepsilon_{0}\approx3\times10^{-3}$ for
$p=0.35$); $O(\xi^{2Q})$ denotes the
finite-iteration residual, below $10^{-8}$ for $\xi\approx0.1$ and $Q=4$;
and $O(1/(N\rho))$ is a finite-SNR term, arising from the
amplitude nonlinearity of $|X|$, that decays as the SNR grows and
accounts for the mild overall rise of the measured efficiency as the
threshold is approached (worst values over the whole cell, $1.02$ at $30$~dB,
$1.07$ at $-5$~dB, and $1.10$--$1.15$ at $-8$~dB;
Section~\ref{sec:experiments}). The refinement samples the DTFT at the
interpolated fractional bin, so no scalloping loss enters, unlike in
the acquisition stage; the relevant coherent SNR here is
$N\rho$, distinct from $\Gamma=NL^{2}\rho$ at the acquisition stage of
Theorem~\ref{thm:threshold}.
The deterministic single-step residual of Lemma~\ref{lem:inverse} is a
separate \emph{bias} of $O(1/N^{2})$ bins, with a squared
contribution to the ratio growing as $O(\rho/N^{3})$; at $N=256$ the
residual reaches $1\%$ of the bound only beyond $+56$~dB. In
particular, the
fixed-point variance of the frequency axis is, in closed form,
\begin{align}
  \operatorname{var}(\hat\nu_{c})
   &= \frac{p^{2}\,\bigl(N-D(2p)\bigr)}
          {4N^{2}\rho\,\bigl[D(p)-N\cos\pi p\bigr]^{2}},
  \label{eq:kappa}\\
  \kappa_{\nu}(p,N)
   &= \frac{\operatorname{var}(\hat\nu_{c})}{\mathrm{CRB}_{\nu_{c}}},
  \nonumber
\end{align}
in which the right-hand side contains no $\delta$: the efficiency is
analytically independent of the cell position. Numerically
$\kappa_{\nu}(0.35,256)=1.003$.
The uniformity transfers to the native parameterization through the
linear map \eqref{eq:nuc}:
$\eta_{\nu}^{\mathrm{native}}
 =[\operatorname{var}(\hat\nu_{c})+c^{2}\operatorname{var}(\hat\mu)]
  /[\mathrm{CRB}_{\nu_{c}}+c^{2}\mathrm{CRB}_{\mu}]\approx1$
when $\eta_{\nu_{c}}\approx1$, the chirp-rate axis is likewise close to the bound (Remark~\ref{rem:muaxis}), and the centered estimates
are uncorrelated (Pillar~A of the proof)---all three observed
in simulation.
\end{theorem}

\begin{IEEEproof}
See Appendix~C of the supplementary material.
\end{IEEEproof}

\begin{remark}[Chirp-rate axis]\label{rem:muaxis}
The $\mu$ axis follows the same centered fixed-point architecture, and
Pillar~A, the decoupling of the two axes, holds for both axes
independently of the cell position. Its parabolic-vertex three-point interpolation is \emph{not} the algebraic inverse of the Fresnel-type
peak curve along $\mu$, so no counterpart of
Lemma~\ref{lem:inverse} (noiseless inversion) exists; at the
fixed point the curve is even in $\delta_{\mu}$, the
vertex is unbiased, and the linearization of Pillar~B carries over.
With $F(s)=\sum_{n}e^{-j\pi(s/N^{2})(n-c)^{2}}
=G(s)e^{j\varphi_{s}}$ the Fresnel sum of the residual chirp, the
fixed-point variance is, in closed form,
\begin{equation}\label{eq:kappa-mu}
  \operatorname{var}(\hat\delta_{\mu})
   = \frac{p^{2}\bigl[N-G(2p)\cos(2\varphi_{p}-\varphi_{2p})\bigr]}
          {16\rho\,\bigl[N-G(p)\bigr]^{2}}
   \quad\text{cells}^{2},
\end{equation}
in which, as in \eqref{eq:kappa}, no step involves $\delta_{\nu}$ or
$\delta_{\mu}$: the constant is cell-position independent.
Numerically $\kappa_{\mu}(0.35,256)
=\operatorname{var}(\hat\mu)/\mathrm{CRB}_{\mu}=0.998$, unity to
within the linearization order, and a single-update Monte Carlo
check at the fixed point gives $1.004$--$1.005$ ($10^{5}$ trials).
\end{remark}

\subsection{Robustness to Cubic Phase Mismatch}\label{ssec:jerk}

Finally we consider a departure from the second-order model
\eqref{eq:model}: the true phase contains a cubic term not included in the model, $\psi[n]=\phi_{0}+2\pi(\nu n+\tfrac12\mu n^{2}
+\tfrac{\zeta}{6}n^{3})$, with jerk $\zeta$ (cycles/sample$^{3}$),
while the estimator continues to fit two phase derivatives. Under a
cubic perturbation, the instantaneous frequency varies along the
signal and the native parameters no longer characterize it
globally, so bias is referenced to the instantaneous parameters at the signal center
$\nu_{c}^{\mathrm{inst}}=\nu+\mu c+\tfrac{\zeta}{2}c^{2}$ and
$\mu_{c}^{\mathrm{inst}}=\mu+\zeta c$---the quantities the centered
estimator reports.

\begin{proposition}[Jerk-induced bias and validity domain]\label{prop:jerk}
At high SNR, under the least-squares phase-projection analysis, the omitted cubic term induces, to first order in $\zeta$, the
deterministic biases
\begin{equation}\label{eq:jerk-bias}
  \mathrm{bias}(\hat\nu_{c})
    = \frac{\zeta\,(3N^{2}-7)}{120}
    \approx \frac{\zeta N^{2}}{40},
  \qquad
  \mathrm{bias}(\hat\mu) = 0,
\end{equation}
i.e., the entire first-order effect of the jerk falls on the frequency
axis; the centered chirp-rate estimate is insensitive to first order. The
second-order model remains adequate as long as the bias is below the
noise floor, which bounds the admissible jerk by
\begin{equation}\label{eq:jerk-domain}
  \zeta_{\max}(\rho)
   = \frac{120}{3N^{2}-7}\,\sqrt{\mathrm{CRB}_{\nu_{c}}(\rho)} .
\end{equation}
\end{proposition}

\begin{IEEEproof}
See Appendix~D of the supplementary material.
\end{IEEEproof}

\begin{remark}[Scope of the robustness claim]\label{rem:jerk-honesty}
Two points limit the scope of the claim. First, the bias \eqref{eq:jerk-bias} is
that of the least-squares projection onto the second-order model,
which a second-order fit incurs, to first order in $\zeta$; the proposed
estimator tracks it to within $2.4$--$3.7\%$
(Section~\ref{sec:experiments}), and its validity domain matches that
of other second-order methods. Specific to the proposed
estimator is the centering. On the uncentered grid, the odd moments do not
vanish, and the jerk bias contaminates the chirp-rate axis as well.
The first-order immunity of $\hat\mu$ in \eqref{eq:jerk-bias} is thus
a third consequence of the centered coordinates, after the CRB
decoupling (Section~\ref{ssec:crb}) and the refinement decoupling
(Theorem~\ref{thm:edgefree}). Second, since the comparison of
admissible jerk ranges across published methods depends on their
respective signal lengths and sampling rates, \eqref{eq:jerk-domain}
is stated in normalized units only.
\end{remark}

\section{Numerical Results}\label{sec:experiments}

Unless stated otherwise, all experiments use $N=256$ samples, a signal
placed at $k_{0}=\operatorname{round}(0.2N)$ frequency bins and
$\ell_{0}=2$ chirp-rate cells, and the residual pair
$(\delta_{\nu},\delta_{\mu})$ of the \emph{centered} parameterization,
$\nu_{c}=(k_{0}+\delta_{\nu})/N$ and $\mu=(\ell_{0}+\delta_{\mu})/N^{2}$
(Section~\ref{ssec:cell}), swept as indicated. The proposed
estimator runs with the fixed configuration of
Section~\ref{sec:estimator} ($p=0.35$, $Q=4$,
$M_{\mathrm{fac}}=3$). The chirp-rate prior is
$\mu_{\max}=5\times10^{-4}$ (about $66$ chirp-rate cells wide at
$N=256$), sampled by a $K_{\mu}=64$-branch dechirp bank at
approximately one-cell spacing; the frequency prior is the full
bandwidth ($B_{\nu}=1$). All
methods search the same prior region and are run on the same noise
realizations.

Five baselines are compared.
\emph{(a)~A\&M-kernel variant and (b)~QSE kernel:} both run
in the same framework as the proposed estimator (the centering, the coarse estimation stage, and the chirp-rate three-point interpolation) and differ from it only in the
frequency kernel, so the comparison isolates the kernel alone.
Variant~(a) uses the half-bin complex-ratio update of
\cite{AboutaniosMulgrew2005}; variant~(b) uses the $q$-shift
interpolator of Solak \emph{et al.} \cite{Solak2022QSE2}, the
two-dimensional extension of the 1-D QSE \cite{Serbes2019QSE},
implemented from their published update and normalization
constant $c(q)$. The original two-dimensional implementation as published, uncentered,
unpadded, and with the chirp rate frozen at the coarse-grid resolution
$1/N^{2}$, is examined separately in Section~\ref{ssec:exp-kernel} as a
control. \emph{(c)~FrFT-GSS:} golden-section search over the chirp rate
\cite{AldimashkiSerbes2020} followed by iterative A\&M frequency
estimation on the dechirped signal. Two implementations are used:
the tolerance-stopped search of the published protocol, and a
strengthened variant that initializes the search from a $64$-point
pre-grid and evaluates the objective through parabolically refined
DTFT peaks on the unpadded $N$-point spectrum, which makes the
search converge reliably; the strengthened variant supplies the
accuracy curves, and
Section~S-II of the supplementary material accounts for both. \emph{(d)~2-D ML oracle:} Nelder--Mead
maximization of the two-dimensional periodogram initialized at the
true parameter values; it attains the CRB at both cell positions
from $-8$~dB upward and is represented there by the $\eta=1$ line
(below $-13$~dB even this reference departs from the bound); being initialized
with unavailable information, it is a conditional reference, not an
implementable competitor. \emph{(e)~DPT:} the discrete
polynomial-phase transform of Peleg and Friedlander
\cite{PelegFriedlander1995} with the variance-optimal delay
$\tau=N/2$; both the DPT$_{2}$ peak and the dechirped
spectrum peak are refined by local maximization of the DTFT
magnitude, and a
noiseless self-check reproduces their variance expressions
(their Eqs.~(42) and~(44)) to within $6\%$.

The
modified DCFT
(MDCFT) estimator of Song \emph{et al.} \cite{SongMDCFT2019} was also
implemented. The published method locates the quasi-peak on an
$N\times N$ transform grid and refines it by iterated interpolation
of MDCFT coefficients at half-bin offsets; our implementation keeps
the grid search and replaces the interpolation by a local ML
refinement of the peak. It
reaches the bound at the inner cell positions
($\eta_{\nu_{c}}=1.01$--$1.03$ at $\delta\in\{0,0.25\}$, $-5$~dB)
and, excluding outlier trials, $1.02$ at the cell corner. The corner
outlier rate is $1.1\times10^{-3}$, an order of magnitude above the
padded acquisition methods, consistent with the scalloping loss of its
unpadded $N\times N$ grid (Theorem~\ref{thm:threshold}). Its grid
cost is $O(N^{2}\log N)$, and it is omitted from the figures for
clarity.

All methods are evaluated on the common task of estimating the frequency at the signal center. The reported efficiency is
$\eta_{\nu_{c}}
 =\mathbb{E}\bigl[(\hat\nu_{c}-\nu_{c})^{2}\bigr]/\mathrm{CRB}_{\nu_{c}}$
with $\hat\nu_{c}=\hat\nu+c\hat\mu$, the MSE convention of
\eqref{eq:eta}, so that uncentered baselines are
measured against the same decoupled bound
(Section~\ref{ssec:crb}). All Monte Carlo experiments use
MC $=2\times10^{4}$ per point (the dedicated outlier count of
Section~\ref{ssec:exp-threshold} uses $3\times10^{4}$);
paired (common) noise realizations are used across methods. Where
method differences are at
the percent level, paired $95\%$ bootstrap confidence intervals
over the common-noise trials are quoted.

\subsection{Verification of Cell-Uniform Efficiency}\label{ssec:exp-uniform}

\begin{figure}[t]
  \centering
  \includegraphics[width=\columnwidth]{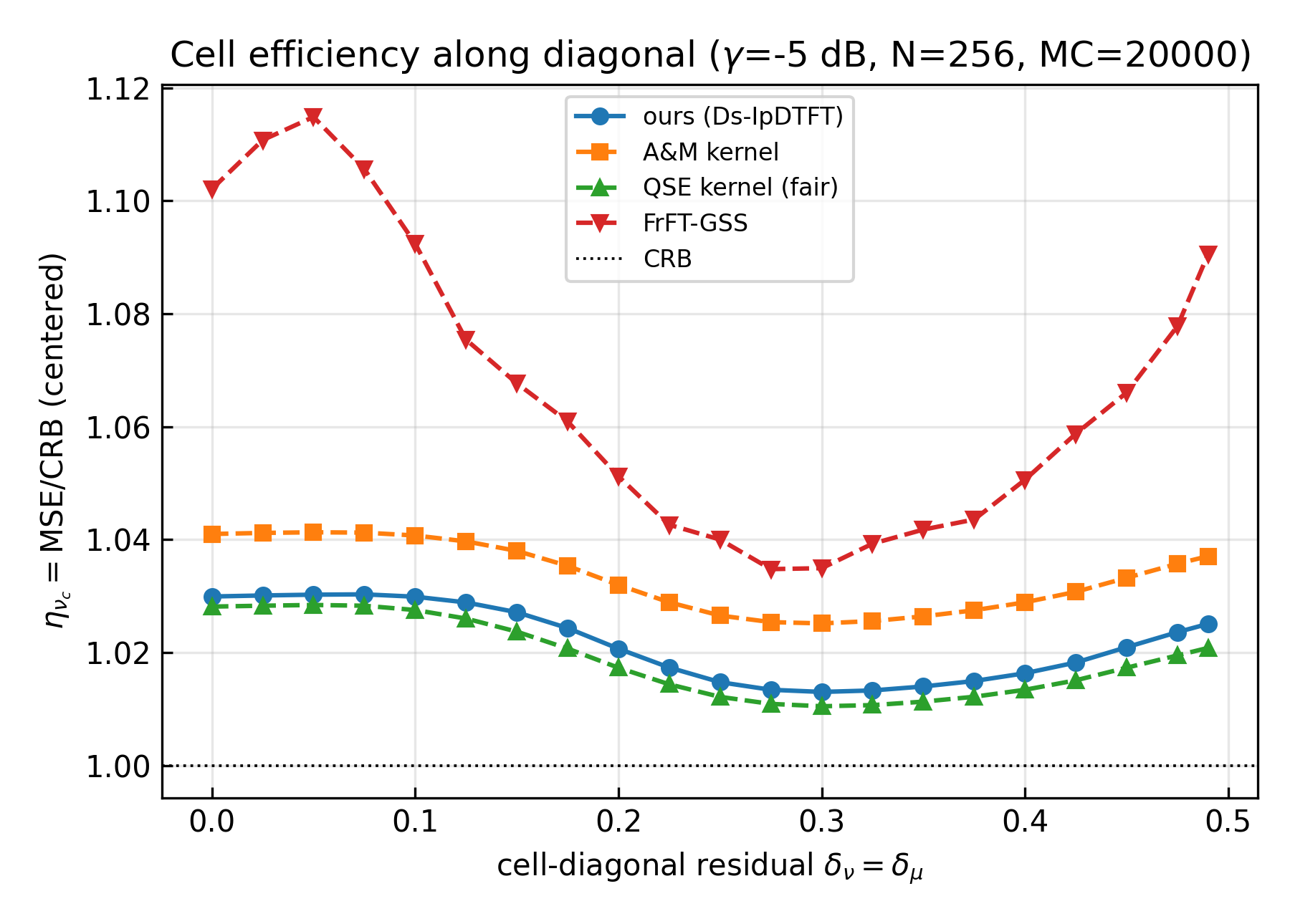}
  \caption{Efficiency $\eta_{\nu_{c}}$ along the cell diagonal
  $\delta_{\nu}=\delta_{\mu}$ at $\gamma=-5$~dB ($N=256$,
  MC $=2\times10^{4}$). The three kernels embedded in the
  same framework are
  flat and near the bound over the entire diagonal (the proposed
  estimator at $1.01$--$1.03$, the QSE kernel at $1.01$--$1.03$,
  and the A\&M-kernel variant at $1.03$--$1.04$), while FrFT-GSS, which
  inherits an uncentered frequency stage, sits $3.5$--$11.5\%$ above the
  bound. They are nearly indistinguishable, locating the uniformity in
  the centered, padded framework.}
  \label{fig:cell-diag}
\end{figure}

Fig.~\ref{fig:cell-diag} shows the cell-diagonal sweep at
$\gamma=-5$~dB, only $3$--$3.5$~dB above the threshold. The proposed
estimator is flat at $\eta_{\nu_{c}}=1.01$--$1.03$ from the cell
center to $(0.49,0.49)$. The uniformity asserted by
Theorem~\ref{thm:edgefree} thus holds at low SNR, where edge effects are
most damaging \cite{Wei2023DsIpDTFT}. The two other kernels embedded
in the same framework, the A\&M variant ($1.03$--$1.04$) and the QSE kernel
($1.01$--$1.03$), track it closely; the differences are $+1.1\%$ to $+1.2\%$ (A\&M)
and $-0.2\%$ to $-0.4\%$ (QSE) with $95\%$ CIs excluding zero, the
direction predicted by Remark~\ref{rem:p}. Quantitatively, evaluating the closed form
\eqref{eq:kappa} at $p=q=N^{-1/3}\approx0.157$ predicts a variance
$0.30\%$ below that at $p=0.35$, in agreement with the measured
difference: it indicates that the margin between the QSE
kernel and the proposed estimator is the effect of the smaller interpolation offset (Remark~\ref{rem:p}), not a property of the kernel form. This behavior stems from two
properties of the framework: the
padded coarse estimation stage confines the starting residual to
$|\delta|\le1/6$, where all three kernels are close to the bound
(Section~\ref{ssec:exp-kernel}); and the centering decouples
the axes, so a chirp-rate residual leaves the centered peak in
place and enters the native frequency only as $-c\,\Delta\mu$ to
leading order, which the mapping
$\hat\nu_{c}=\hat\nu+c\hat\mu$ removes.
FrFT-GSS, by contrast, lies $3.5$--$11.5\%$ above the bound, reflecting the
uncentered A\&M frequency stage it inherits. DPT is below
its breakdown threshold at this SNR
(Section~\ref{ssec:exp-threshold}) and is omitted here.

We compare the efficiency of the proposed estimator for each value
of the interpolation factor $p\in[0.05,0.49]$, at four SNRs from
$-7$~dB to $+30$~dB and at three cell positions; the resulting
estimators share the same coarse estimation stage, so their
differences lie in the fine estimation stage. If the margin between
the QSE kernel and the proposed estimator is the effect of the
smaller interpolation offset, the efficiency of these estimators
must follow $\kappa_{\nu}(p)$ (Fig.~S5 of the supplementary
material). The
measured curves track $\kappa_{\nu}(p)$: the prediction falls inside
the $95\%$ CI at $111$ of
the $120$ grid points, the total spread is $1.0$--$1.5\%$ against the
predicted $1.3\%$, and the near-threshold curves show no additional
$p$-dependence relative to $+30$~dB. At $p=q=N^{-1/3}$ the paired
change is $-0.14\%$ to $-0.37\%$, bracketing both the $-0.30\%$
prediction and the measured QSE margin above; on the chirp-rate axis
all differences stay below $0.02\%$. This confirms the
$p$-insensitivity asserted in Remark~\ref{rem:p}: over the range
$p\in[0.05,0.49]$, the loss of the fixed $p=0.35$ with respect to
the best value of $p$ is below $0.4\%$ in variance at all four SNRs
and three cell positions. The estimators differ instead in uniformity and
threshold behavior, examined next.

\begin{figure*}[t]
  \centering
  \includegraphics[width=\textwidth]{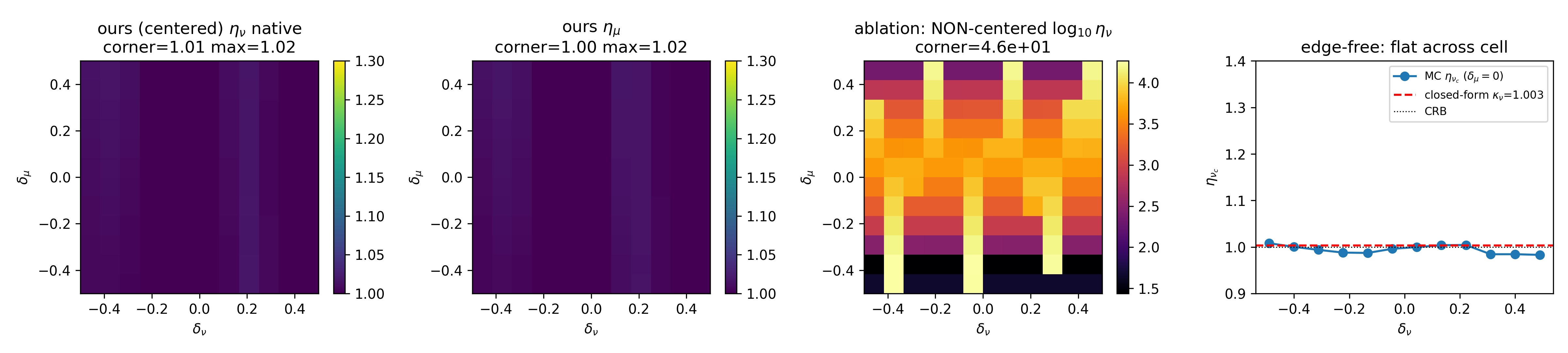}
  \caption{Uniformity over the full $12\times12$ residual cell
  ($N=256$, $30$~dB, MC $=2\times10^{4}$). Left to right: native
  $\eta_{\nu}$ of the proposed estimator (max $1.02$, corners
  included); $\eta_{\mu}$ (max $1.02$); ablation, the identical
  kernel and schedule without centering, shown as
  $\log_{10}\eta_{\nu}$ (range $27$--$1.8\times10^{4}$); and the
  $\delta_{\mu}=0$ slice of $\eta_{\nu_{c}}$ against the closed-form
  fixed-point efficiency $\kappa_{\nu}=1.003$ of
  Theorem~\ref{thm:edgefree} (dashed), which carries no
  $\delta$-dependence and no fitted correction.}
  \label{fig:edgefree}
\end{figure*}

Fig.~\ref{fig:edgefree} validates Theorem~\ref{thm:edgefree} over the
full cell. The efficiency surfaces of the proposed estimator
are flat ($\eta_{\nu}\le1.02$ and $\eta_{\mu}\le1.02$ across the
residual cell), and the measured surface agrees with the
analytical constant $\kappa_{\nu}=1.003$ of \eqref{eq:kappa} with
mean deviation $1.0\%$.
The third panel is the centering ablation (Remark~S2 of the supplementary material): with
centering removed and everything else identical,
$\eta_{\nu}$ spans $27$ to $1.8\times10^{4}$ across the cell. The
best uncentered cell is still $27\times$ worse than the worst
centered one.

The closed forms of Theorem~\ref{thm:edgefree} and
Lemma~\ref{lem:inverse} carry over to other signal lengths
(Fig.~S6 of the supplementary material): the fixed-point efficiency $\kappa_{\nu}$
converges to a small $N$-independent constant
($\kappa_{\nu}-1\approx3\times10^{-3}$), and the
noiseless single-step inversion residual decays as $O(1/N^{2})$ over
$N=64$--$4096$---below $4\times10^{-6}$ at $N=256$, negligible
against any noise floor, and driven to machine precision by the
subsequent iterations.

\begin{figure*}[t]
  \centering
  \includegraphics[width=\textwidth]{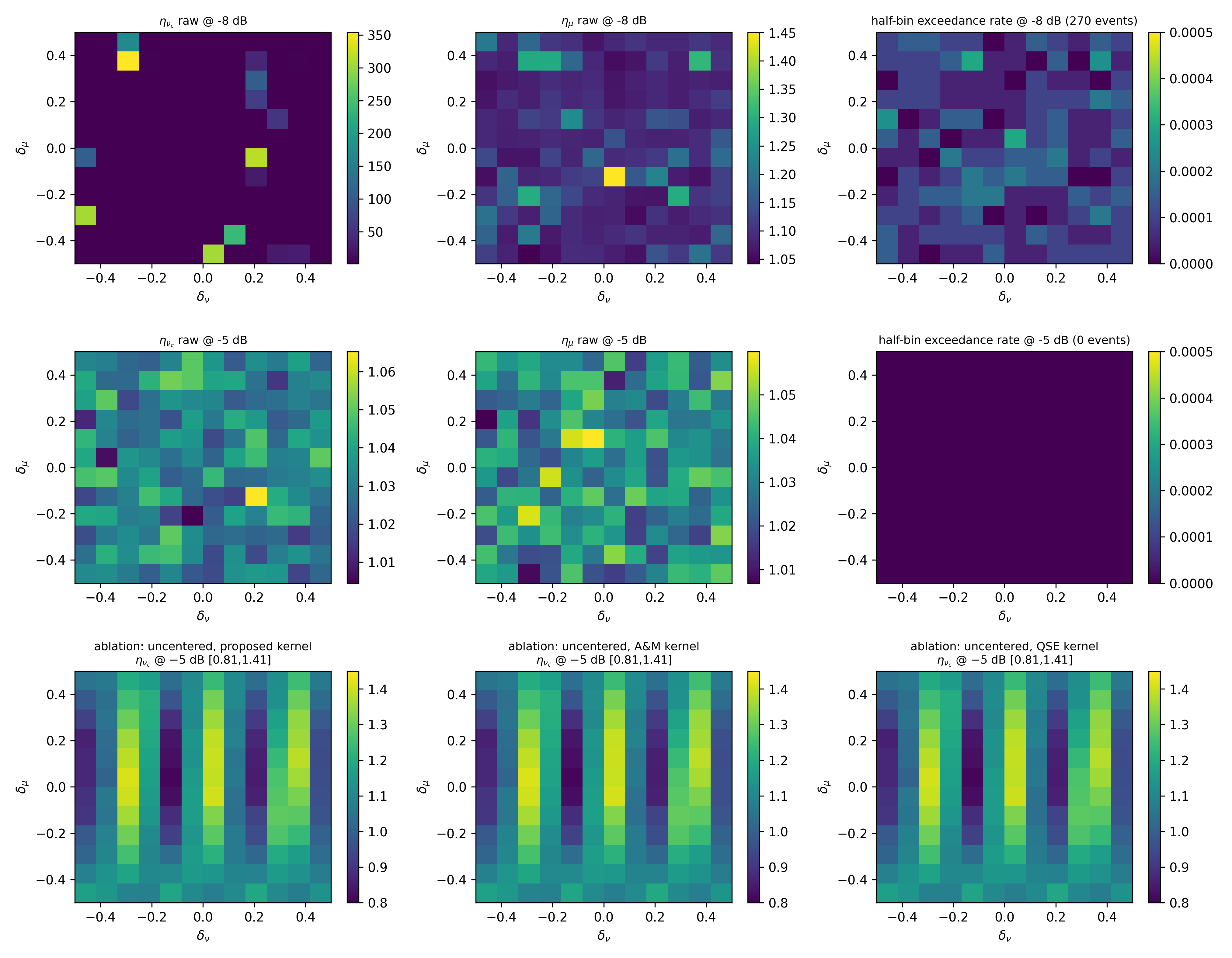}
  \caption{Uniformity over the whole cell at $-8$ and $-5$~dB and the
  centering ablation ($N=256$,
  $12\times12$ grid over $[-0.49,0.49]^{2}$, MC $=2\times10^{4}$ per
  cell, disjoint noise seeds per cell). Rows 1--2, proposed
  estimator at $-8$~dB ($0.5$~dB above the center onset, at the
  corner onset) and $-5$~dB: $\eta_{\nu_{c}}$,
  $\eta_{\mu}$, and the fraction of trials with a centered-frequency
  error exceeding half a bin. At $-5$~dB, no such event occurs in
  $2.88\times10^{6}$ trials and the surfaces are flat: median $1.03$
  on both axes,
  worst of $144$ cells $1.07$ ($\nu_{c}$) and $1.06$ ($\mu$). At $-8$~dB, $270$ outlier events ($9.4\times10^{-5}$, over $128$ cells) dominate the raw
  MSE of the cells
  they fall in; excluding them, the median is $1.07$ and the worst
  cell $1.10$ ($\nu_{c}$) and $1.15$ ($\mu$). Row 3, the centering
  ablation at $-5$~dB: removing the centering (all else identical)
  produces the \emph{same} surface for all three kernels---proposed,
  A\&M, and QSE agree cell by cell to within $0.7\%$ on
  $\eta_{\nu_{c}}$ and $3.3\%$ on $\eta_{\mu}$
  ($\eta_{\nu_{c}}$ medians $1.11$/$1.11$/$1.11$, oscillating over
  $[0.81,1.41]$; $\eta_{\mu}$ medians $2.23$/$2.24$/$2.23$)---while
  their centered counterparts are all flat (worst values $1.07$/$1.07$/$1.06$). Both signs of both
  residuals are covered.}
  \label{fig:lowsnr-surface}
\end{figure*}

The $30$-dB surfaces of Fig.~\ref{fig:edgefree} isolate the
refinement stage; the first two rows of
Fig.~\ref{fig:lowsnr-surface} repeat the sweep over the whole cell at $-8$ and $-5$~dB. At
$-5$~dB the surfaces are flat over the entire cell, including the
negative-residual quadrants: median $\eta_{\nu_{c}}$ and $\eta_{\mu}$
both $1.03$, worst of the $144$ cells $1.07$, and no
centered-frequency error exceeded half a bin in $2.88\times10^{6}$
trials.
At $-8$~dB---$0.5$~dB above the center onset and essentially at the
corner onset---rare
acquisition near-misses appear ($270$ half-bin outlier events, $9.4\times10^{-5}$) and dominate the MSE of the cells they fall in;
excluding them, the simulation gives a median of $1.07$ and a worst value of $1.10$ ($\nu_{c}$) and $1.15$ ($\mu$). The uniformity of
Theorem~\ref{thm:edgefree} and Remark~\ref{rem:muaxis} thus holds at
the target SNRs, with only rare outlier events near threshold.

The third row of Fig.~\ref{fig:lowsnr-surface} extends the centering
ablation (Remark~S2 of the supplementary material) to the other
kernels. They
produce the same degradation cell by
cell (to within $0.7\%$ on $\eta_{\nu_{c}}$, $3.3\%$ on
$\eta_{\mu}$): the chirp-rate MSE roughly doubles
($\eta_{\mu}$ median ${\approx}2.2$, up to $5.2$), and
$\eta_{\nu_{c}}$ oscillates over $[0.81,1.41]$ about unity. The
sub-unity values are inherent to a biased estimator.
Without centering, the chirp-rate refinement acquires a
small systematic offset. Their centered
counterparts are all flat. The uncentered penalty
grows with SNR: at $-5$~dB it roughly doubles the chirp-rate
MSE, while at $30$~dB it reaches
$\eta_{\nu}\in[27,1.8\times10^{4}]$, as shown in
Fig.~\ref{fig:edgefree}.
Uniform efficiency in the joint problem is thus a property of the
centering, not of the interpolation kernel.

\subsection{Verification of the Threshold Characterization}\label{ssec:exp-threshold}

\begin{figure*}[t]
  \centering
  \includegraphics[width=\textwidth]{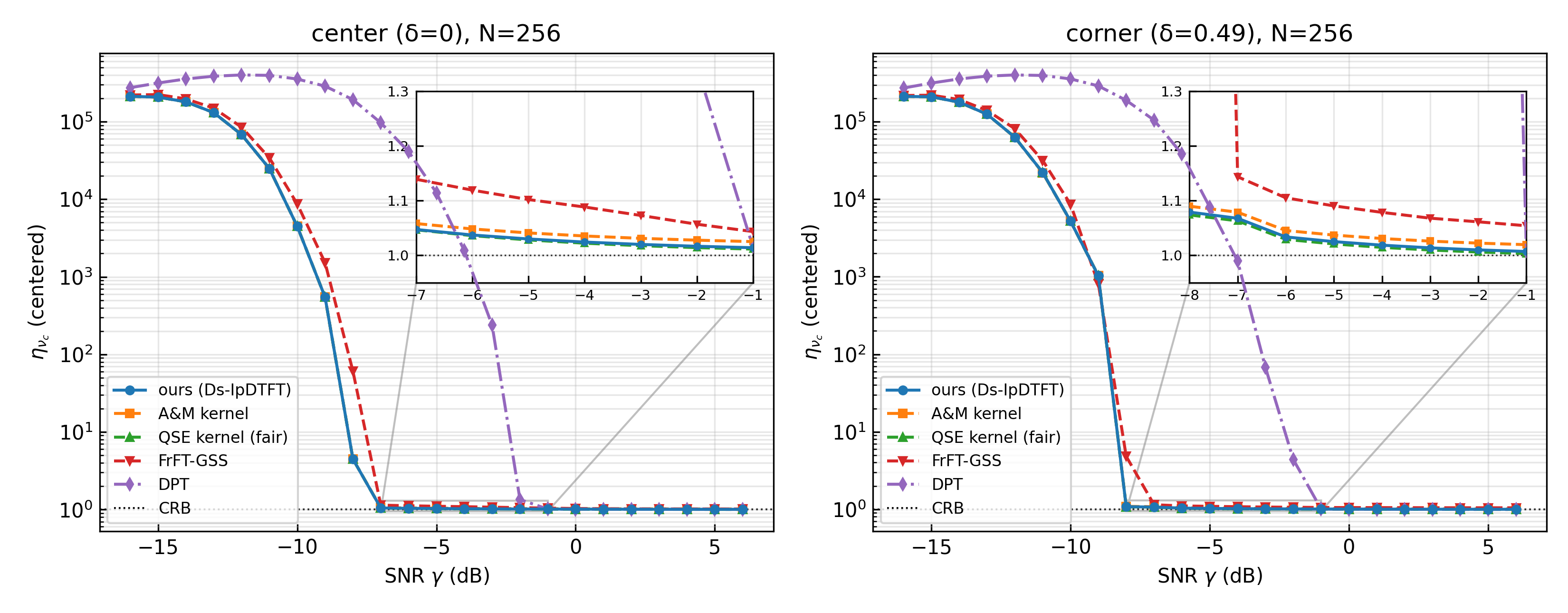}
  \caption{Efficiency versus SNR at the cell center (left) and corner
  (right) ($N=256$, MC $=2\times10^{4}$ per point). The four
  acquisition-based methods leave the
  prior floor between $-10$ and $-8$~dB and follow the bound up to
  $+6$~dB at both cell positions, with closely matched thresholds at
  the center and the corner, the cell-position independence that the
  prescribed padding secures (Theorem~\ref{thm:threshold}); the
  phase-based DPT leaves the floor only at about $-2$~dB and
  approaches the bound above its own threshold. Inside the
  transition, at $-9$~dB, the acquisition-based methods are at
  ${\sim}10^{3}$ (outlier-dominated MSE); by $-8$~dB, at the corner, the three estimators with a fixed operation count are within $10\%$ of the bound while
  FrFT-GSS remains at ${\approx}4.8$ (its higher
  near-threshold outlier rate, quantified in the text); the $-8$~dB
  center values remain outlier-dominated, cf.\ the exceedance maps
  of Fig.~\ref{fig:lowsnr-surface}. Insets: linear-scale zoom of the
  transition exit, $-7$ to $-1$~dB at the center and $-8$ to $-1$~dB
  at the corner.}
  \label{fig:rmse-snr}
\end{figure*}

\begin{figure}[t]
  \centering
  \includegraphics[width=\columnwidth]{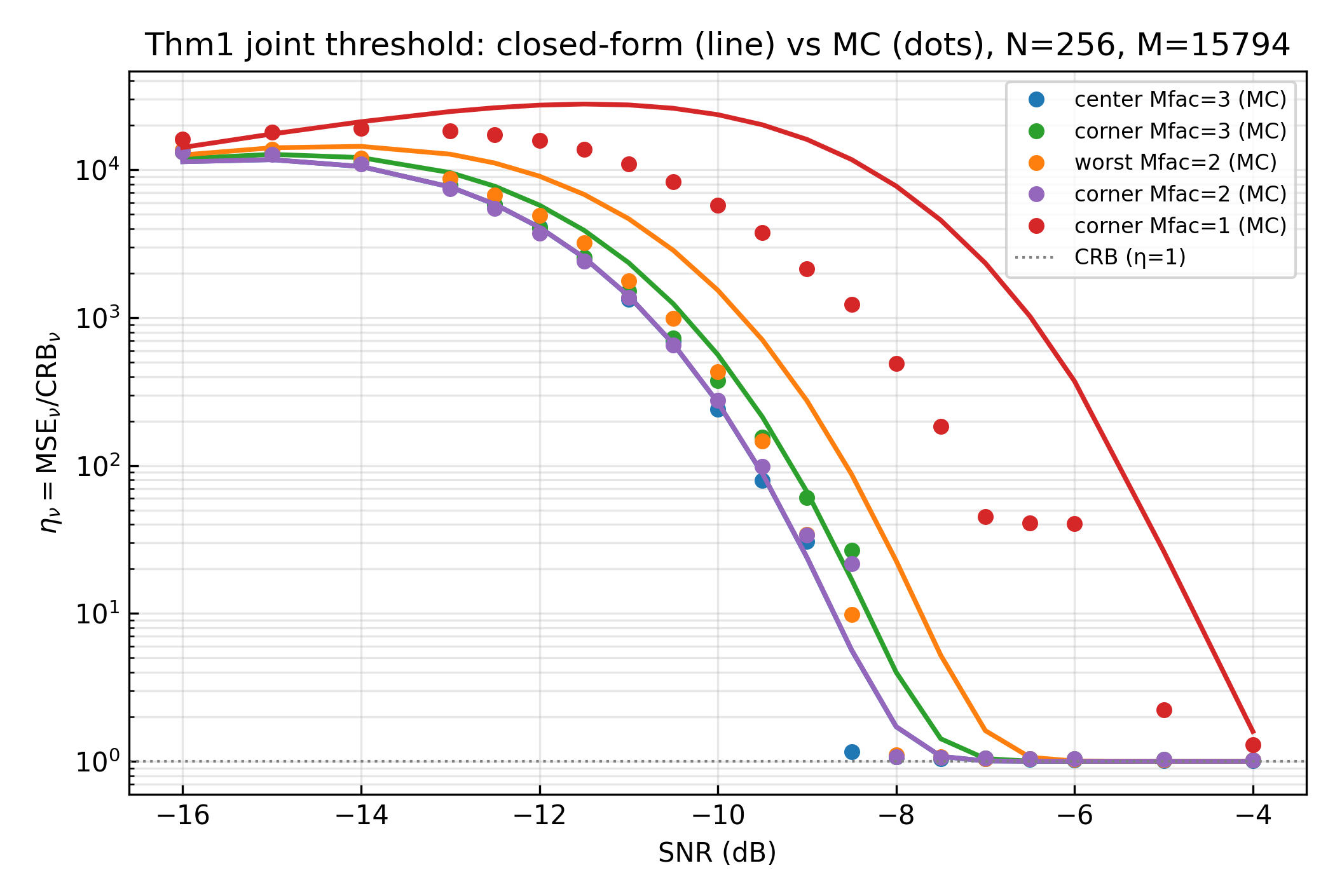}
  \caption{Theorem~\ref{thm:threshold} against Monte Carlo
  ($N=256$, MC $=2\times10^{4}$ per point). Lines: closed form
  \eqref{eq:three-segment}--\eqref{eq:gth} with the single fitted
  constant $\mathcal{M}=15794$ (calibrated at the center, $M_{\mathrm{fac}}=3$);
  dots: measured $\mathrm{MSE}_{\nu}/\mathrm{CRB}_{\nu}$. The four
  remaining curves are independent predictions through the scalloping
  loss \eqref{eq:straddle}: the measured onset
  advances at reduced padding and off-center positions track the
  predicted scalloping shifts, quantified against the higher-precision
  onset measurement in the text.}
  \label{fig:threshold}
\end{figure}

Fig.~\ref{fig:rmse-snr} sweeps the SNR from $-16$ to $+6$~dB at the
cell center and corner. With padding, the threshold of
the proposed estimator is nearly the same at the center and the corner
(${\approx}{-8.5}$ and ${\approx}{-8.0}$~dB), and the cell-position
dependence is predicted by
Theorem~\ref{thm:threshold} via the scalloping loss. Above
threshold, the proposed estimator, its A\&M and QSE variants, and
FrFT-GSS track the bound on the centered-frequency
axis to $+6$~dB (FrFT-GSS with the $3.5$--$11.5\%$ excess of
Fig.~\ref{fig:cell-diag}), apart from the rare outlier
excursions. The phase-based DPT breaks
down much earlier: it deviates from the bound at
${\approx}{-2.0}$~dB (center) and ${\approx}{-1.3}$~dB
(corner), some
$6$--$7$~dB above the other methods, although it reaches the bound
above its own threshold.

The near-threshold outlier statistics of the coarse estimation stage are
consistent with the prediction of Theorem~\ref{thm:threshold}. At the
worst case, the $(0.49,0.49)$ corner at $-8$~dB, just above
threshold, a direct count over $3\times10^{4}$ trials gives an
outlier probability (defined as a centered-frequency error exceeding
half a bin) of $2.3\times10^{-4}$ for FrFT-GSS and $1.0\times10^{-4}$
for each of the three estimators with a fixed operation count
($95\%$ Clopper--Pearson intervals $[0.9,4.8]\times10^{-4}$ and
$[0.2,2.9]\times10^{-4}$, respectively); an exact test on the paired
counts ($7$ versus $3$ events, two-sided $p\ge0.13$) does not reject
equal rates, and no method is
outlier-free near threshold. The determinism of the proposed schedule
fixes the operation count, not the outlier probability (Section~S-II
of the supplementary material).

Fig.~\ref{fig:threshold} tests the closed forms of
Theorem~\ref{thm:threshold}: the lines are the three-segment form
\eqref{eq:three-segment}--\eqref{eq:gth} with the single constant
$\mathcal{M}=15794$ calibrated once on the center
($M_{\mathrm{fac}}=3$) curve, under the protocol of
Remark~\ref{rem:thm1-honesty}; everything else is prediction. The
plotted ratio is the \emph{native}
$\mathrm{MSE}_{\nu}/\mathrm{CRB}_{\nu}$ \eqref{eq:crb-native}, the quantity at the acquisition stage; for the centered estimator
$\mathrm{MSE}_{\nu}=\mathrm{MSE}_{\nu_{c}}+c^{2}\,\mathrm{MSE}_{\mu}$,
so the three-segment form carries over, with the floor set by
$\Sigma_{\nu}=B_{\nu}^{2}/12$ (the $c^{2}B_{\mu}^{2}/12$ contribution
is below $2\%$, and $\eta$ floors are ${\approx}16\times$ lower
than in the centered convention of Fig.~\ref{fig:rmse-snr}).

Three features of the figure carry the verification. First,
$\Sigma_{\nu}=1/12$ predicts a no-information floor of
$\eta=1.46\times10^{4}$,
against $1.34\times10^{4}$ measured at $-16$~dB. Second, the predicted
onset of the threshold region (the SNR at which $\eta$ crosses $2$
from above) is higher than the measured onset by $0.5$~dB on the
calibration curve and by $-0.0$ to $+1.0$~dB on the four configurations not used in the calibration, consistent with
Remark~\ref{rem:thm1-honesty}; the measured
crossings themselves carry
$0.4$--$0.95$~dB bootstrap intervals, because near threshold the
MSE is dominated by rare outliers. Third, the cell-position dependence
follows the predicted scalloping shifts: the unpadded corner (scalloping
loss $10\log_{10}\operatorname{sinc}^{2}(0.49)\approx-3.8$~dB)
advances the measured onset by $3.8$~dB, and the padded
configurations advance it by
$0.46$--$0.50$~dB with $0.5$--$0.95$~dB confidence intervals,
bracketing the predicted $0$--$0.99$~dB differences, which are not
individually resolved at this precision. One calibrated constant
thus reproduces the floor, the onset, and the position dependence
of all five curves. The cell-position shifts test the $L^{2}$ term in
\eqref{eq:gth} independently of the calibrated constant, and $M_{\mathrm{fac}}=3$ keeps the measured cell-position spread at
${\approx}0.5$~dB, consistent with its $0.4$~dB deterministic
scalloping cap (Section~\ref{ssec:coarse}) to within
onset-estimation noise.

\subsection{Kernel Ablation: Edge Regime and Reduced Padding}
\label{ssec:exp-kernel}

Section~S-IV of the supplementary material isolates the frequency
kernel in the identical centered, padded framework. At the target SNRs the
selectable-$p$ kernel is the better of the two at every tested cell
position ($1.025$ against $1.037$ at the $-5$~dB corner, $1.079$
against $1.090$ at $-8$~dB), but within the padded two-stage estimator the gap is only ${\approx}1\%$, because the zero-padded coarse estimation stage already
confines the refinement residual to $|\delta|\le1/6$, where even the
half-bin samples used in the A\&M kernel remain inside the main lobe.
Two control experiments separate the kernel from the framework. The
selectable-$p$, A\&M and QSE kernels coincide when placed in that same
framework, so the deviation of the QSE2 implementation as published is a property of its structure and not of its kernel. Reducing the padding factor degrades
both kernels together, driven by the coarse-stage scalloping loss and
the incipient outliers of Theorem~\ref{thm:threshold}, so the kernel
choice does not compensate for reduced padding.

\subsection{Comparison with the Closest Prior Joint Interpolator}
\label{ssec:exp-jiang}

The estimator most closely related to ours is that of Jiang and Le
\cite{Jiang2013}, to our knowledge the only prior interpolation
method that operates directly on the dechirped DTFT/likelihood
samples; interpolation-based refinements in transform
domains---FrFT \cite{SongFrFT2013,Liu2018FrFT}, modified DCFT
\cite{SongMDCFT2019}, and ambiguity function \cite{Feng2012DAF}---%
instead recover the parameters through
1-D estimators after the transform-domain search discussed in
Section~\ref{sec:intro}. We implemented the architecture from
their published equations---a coarse trial search over the chirp rate,
a single three-point parabolic interpolation of the log-likelihood on
that axis, and a plug-in single-tone frequency estimator, for which we
use the A\&M estimator \cite{AboutaniosMulgrew2005} as the paper leaves
the inner estimator unspecified. In the noiseless limit, the frequency
is recovered exactly, while the single parabolic pass leaves a
chirp-rate residual of ${\approx}0.06$ cell; the proposed $Q$-iteration
refinement drives the same quantity to machine precision. Under noise,
on the centered-frequency axis, the Jiang--Le implementation reaches efficiency close to the CRB above its threshold but sits uniformly
about $1\%$ above the proposed estimator across $-6$ to $+6$~dB, at
both cell positions (differences $+1.0\%$ to $+1.3\%$, all
$95\%$ CIs excluding zero; $N=256$, MC $=2\times10^{4}$). The substantive distinction is
structural rather than numerical: their chirp-rate axis is localized by
an explicit search, their thresholds are reported only by simulation, and
neither the cell-uniformity nor the threshold behavior is analyzed---%
the gaps this paper addresses.

\subsection{Further Verification}\label{ssec:exp-further}

Three further studies are reported in the supplementary material.
First, at $N=32$, $64$, and $512$ with the
chirp-rate search range fixed in physical units, the center and corner efficiencies of the
proposed estimator and of its A\&M variant stay within $6\%$ of the
bound and the measured threshold onsets at the two cell positions
are within $0.5$~dB of each other, so the cell uniformity of
Theorem~\ref{thm:edgefree} and the threshold equalization of
Section~\ref{ssec:coarse} carry over to these signal lengths; the onset moves
by $11.9$~dB from $N=32$ to $N=512$, against the $12.0$~dB predicted
by the $1/N$ scaling of \eqref{eq:gth} (Section~S-I). Second, the measured
runtimes verify the $O(N\log N)$ scaling for all these estimators, and the two FrFT-GSS implementations bracket
the proposed estimator (Section~S-II). Third, under a cubic phase mismatch the degradation of the frequency estimate follows the closed form
\eqref{eq:jerk-bias} to within $2.4$--$3.7\%$, while the chirp-rate
estimate shows no bias trend over the same range, as
Proposition~\ref{prop:jerk} predicts (Section~S-III).

\section{Conclusion}\label{sec:conclusion}

This paper treated the joint estimation of frequency and chirp rate as
a uniformity problem, because the residual depends on the parameters
to be estimated and a reliable
two-stage estimator must therefore maintain its
accuracy and its threshold at the worst cell position. The
proposed estimator addresses both failure modes by centering the
coordinates and padding the acquisition. The
centering removes the frequency--chirp-rate cross-term
of the Fisher information that couples the two parameters in the natural parameterization; the zero-padded acquisition bank equalizes the detection
threshold across the cell; and the selectable-$p$ kernel, operating on DTFT samples at fractional bins, inverts the Dirichlet kernel
throughout the residual cell and keeps its samples in the main lobe
down to the cell edge, a benefit that is modest within the padded
pipeline. The supporting theory
consists of two closed forms: an MSE and threshold
characterization across the full SNR range with a single calibrated
cell count, and
variance ratios on both axes, independent of
the cell position. A
closed-form jerk-bias analysis shows the chirp-rate estimate to be insensitive, to first order, to a cubic phase mismatch, a property
specific to the centering. Monte Carlo experiments corroborated the
theoretical predictions. The frequency-axis efficiency stays within $7\%$
of the bound over the whole cell at $-5$~dB, which is $3.5$~dB above
the approximately $-8.5$~dB threshold at $N=256$; the threshold
predictions fall within $1.0$~dB on four configurations outside the calibration, from a single calibrated constant. For
low-SNR burst applications such as satellite telemetry and the
detection of highly dynamic targets in radar and sonar systems, the
proposed estimator provides
predictable latency, a schedule fixed in advance, and accuracy
independent of the residual-cell position of the true parameters.

The analysis is confined to a single-component second-order model;
two extensions remain open. Multiple chirp components
motivate an estimate-and-subtract front end, and its interaction with the
threshold theory is open; damped chirps and very short signals change
the spectral kernel and the cancellations on which the
centered analysis rests. Both directions can reuse the framework
developed here: conditional detection probability for the
acquisition stage and fixed-point noise gain for the refinement
stage; both are under investigation.
\emph{Reproducibility.} The code that reproduces every figure
in this paper will be made publicly available upon publication.

\bibliographystyle{IEEEtran}
\bibliography{refs}

\clearpage
\setcounter{section}{0}\setcounter{figure}{0}\setcounter{table}{0}
\setcounter{equation}{0}\setcounter{theorem}{0}\setcounter{lemma}{0}
\setcounter{proposition}{0}\setcounter{corollary}{0}\setcounter{remark}{0}
\setcounter{assumption}{0}
\renewcommand{\thesection}{S-\Roman{section}}
\renewcommand{\thefigure}{S\arabic{figure}}
\renewcommand{\thetable}{S\arabic{table}}
\renewcommand{\theremark}{S\arabic{remark}}
\begin{center}\large\textbf{Supplementary Material}\end{center}

This document contains the proofs of Theorem~1, Lemma~1, Theorem~2 and
Proposition~1 of the main paper, together with three supporting
experiments: validation at other signal lengths, the runtime and
complexity comparison, and the unmodeled-jerk study. Equation, figure
and section numbers without the prefix \textup{S} refer to the main
paper.

\section{Validation at Other Signal Lengths}\label{ssec:exp-nvalid}

\begin{table}[t]
\caption{Validation at other signal lengths ($\mu_{\max}\approx32/N^{2}$,
the fixed-duration convention of
Section~\ref{ssec:exp-complexity};
$\gamma_{\mathrm{op}}=-5+10\log_{10}(256/N)$~dB; MC $=2\times10^{4}$
throughout; $N=256$ row from
Figs.~\ref{fig:cell-diag} and~\ref{fig:threshold}).}
\label{tab:nvalid}
\centering
\begin{tabular}{ccccccc}
\hline
 & & \multicolumn{2}{c}{proposed, $\eta_{\nu_{c}}$}
   & \multicolumn{2}{c}{A\&M variant, $\eta_{\nu_{c}}$}
   & onset (dB) \\
$N$ & $\gamma_{\mathrm{op}}$ & center & corner & center & corner
   & center\,/\,corner \\
\hline
$32$  & $+4$ & $1.02$ & $1.05$ & $1.03$ & $1.06$ & $+0.9$\,/\,$+0.9$ \\
$64$  & $+1$ & $1.04$ & $1.03$ & $1.06$ & $1.05$ & $-2.0$\,/\,$-2.1$ \\
$256$ & $-5$ & $1.03$ & $1.03$ & $1.04$ & $1.04$ & $-8.5$\,/\,$-8.0$ \\
$512$ & $-8$ & $1.02$ & $1.05$ & $1.03$ & $1.06$ & $-11.0$\,/\,$-11.0$ \\
\hline
\end{tabular}
\end{table}

Table~\ref{tab:nvalid} repeats the two headline measurements at
$N=32$, $64$, and $512$ with the chirp-rate search range fixed in physical units,
with the SNR shifted so that the distance to the predicted
threshold is the same at each length. The center and corner
efficiencies of the proposed estimator and of its A\&M variant remain
within $6\%$ of the bound at all three lengths and both cell
positions, and the
measured threshold onsets at the center and the corner are
within $0.5$~dB of each other, so the cell uniformity of
Theorem~\ref{thm:edgefree} and the threshold equalization of
Section~\ref{ssec:coarse} carry over to these signal lengths. The onset
itself moves by $11.9$~dB from $N=32$ to $N=512$, matching the
$12.0$~dB that the
$1/N$ scaling of \eqref{eq:gth} predicts for the $16\times$ length
ratio.

\section{Runtime and Complexity Comparison}\label{ssec:exp-complexity}

\begin{figure}[t]
  \centering
  \includegraphics[width=\columnwidth]{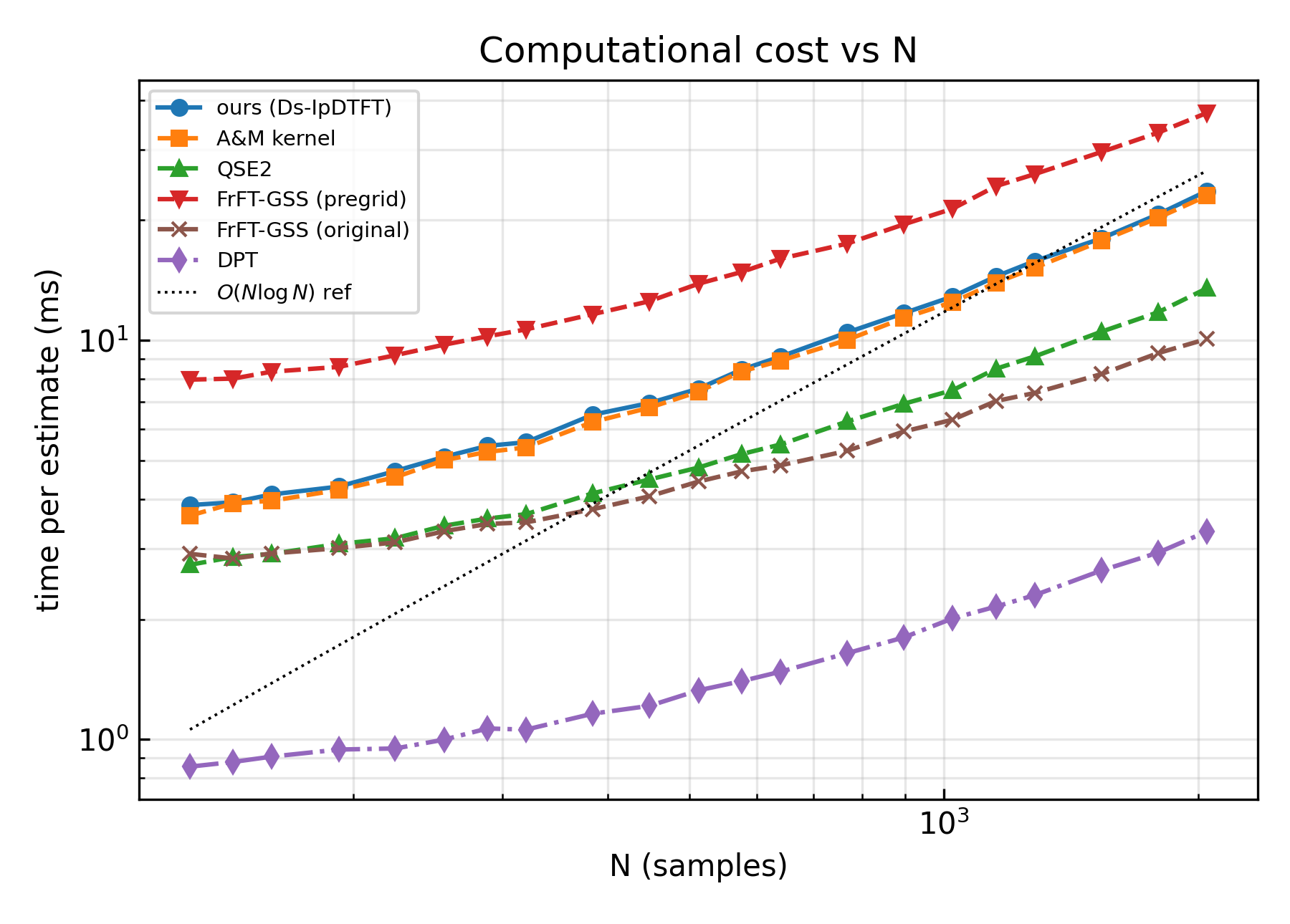}
  \caption{Runtime per estimate versus $N$: per signal the minimum
  of three repetitions is kept (suppressing scheduler jitter), and
  the median across $100$ signals is plotted; identical hardware and
  implementation
  framework; chirp-rate search range fixed in physical units, so the branch count is
  constant across $N$; FFT-friendly lengths $N=2^{a}m$,
  $m\in\{1,3,5,7,9\}$. The proposed
  estimator, its A\&M variant, and the QSE2 transplant scale as
  $O(N\log N)$ (dotted reference, least-squares fit); the two
  FrFT-GSS variants (tolerance-stopped protocol and strengthened
  pre-grid implementation) bracket the
  proposed estimator. DPT is the cheapest
  method tested.}
  \label{fig:complexity}
\end{figure}

Fig.~\ref{fig:complexity} reports measured runtimes, which serve
here to verify scaling: all
methods track the $O(N\log N)$ reference; at small $N$
the measured times lie above it, dominated by the per-call overhead
of the interpreted implementation. The proposed schedule prescribes
the number and set of transform evaluations in advance; the
golden-section search evaluates fewer transforms under its
tolerance-stopped protocol, at the cost of an inherently sequential
count that varies with the realization, while the strengthened
variant trades
additional evaluations for global acquisition. The naive QSE2 without
padding and the chirp-rate refinement is the cheapest of the three
variants, but forfeits threshold
equalization at the cell corner and cell-uniform efficiency
(Sections~\ref{ssec:exp-threshold}
and~\ref{ssec:exp-kernel}). DPT is
the cheapest method tested, at the threshold cost
quantified in Section~\ref{ssec:exp-threshold}. The MDCFT
baseline (not
plotted) is dominated by its $N\times N$ transform grid,
$O(N^{2}\log N)$, an order of magnitude above the fixed-schedule
methods at
$N=1024$. All curves are single-threaded NumPy implementations of
comparable optimization effort on the same workstation (Intel Core
i9-9880H, $32$~GB; Python~3.12, NumPy~2.4;
\texttt{perf\_counter} after one warm-up call per method); absolute
times
are implementation-dependent, and the informative content is the
relative ordering and the scaling with $N$ at the stated prior width.
Fixing the search range in physical units is the fixed-duration convention: the
observation window $T$ is held fixed while the sampling rate grows
with $N$, so
$\dot f_{\max}=\mu_{\max}f_{s}^{2}=32/T^{2}$ is constant and the
branch count is constant across $N$. At a fixed sampling rate, the
same physical range would span $O(N^{2})$ cells and the bank cost
would grow accordingly (Section~\ref{ssec:complexity}). The medians
hide one further
difference: because its algorithmic parameters are fixed,
independent of the parameters to be estimated, the proposed
estimator has constant complexity and identical runtime on each
run, whereas the golden-section search time varies from signal
to signal.

\section{Robustness to Unmodeled Jerk}\label{ssec:exp-jerk}

\begin{figure*}[t]
  \centering
  \includegraphics[width=\textwidth]{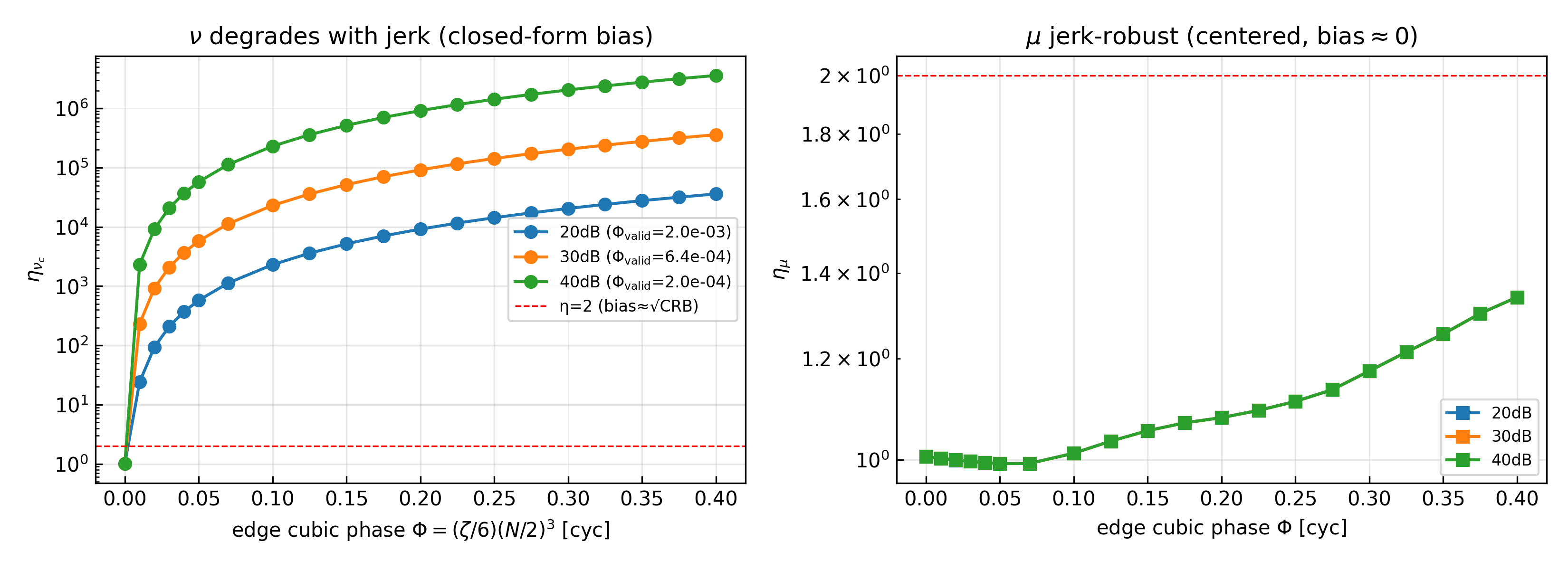}
  \caption{Unmodeled jerk ($N=256$, MC $=2\times10^{4}$),
  parameterized by the
  edge cubic phase $\Phi=(\zeta/6)(N/2)^{3}$ in cycles. Left:
  $\eta_{\nu_{c}}$ degrades with $\Phi$ as the bias closed
  form \eqref{eq:jerk-bias} predicts (deviation
  $2.4$--$3.7\%$ over the sweep), with the validity bound
  \eqref{eq:jerk-domain} marked per SNR
  ($\Phi_{\mathrm{valid}}=2.0\times10^{-3}$, $6.4\times10^{-4}$,
  $2.0\times10^{-4}$ cycles at $20/30/40$~dB). Right:
  $\eta_{\mu}$ stays flat (${\le}1.34$) over the entire sweep: the
  first-order immunity of the centered chirp-rate estimate
  (Proposition~\ref{prop:jerk}).}
  \label{fig:jerk}
\end{figure*}

Fig.~\ref{fig:jerk} examines robustness to unmodeled jerk, sweeping
the cubic term over a $40\times$ range at $20$--$40$~dB; the two
estimates behave as Proposition~\ref{prop:jerk} predicts. The
degradation of the frequency estimate follows the deterministic bias
\eqref{eq:jerk-bias} to within $2.4$--$3.7\%$, while the chirp-rate
estimate shows no bias trend over the
same range. In cell units the measured chirp-rate bias is
consistent with zero to within the Monte Carlo error (below
$6\times10^{-6}$ of a cell over the sweep), while the
frequency-axis bias grows to $0.46$ cell, and
$\eta_{\mu}\le1.34$
throughout, a variance-level excess rather than the
linear-in-$\Phi$ growth of the frequency axis. The validity bound
\eqref{eq:jerk-domain} contracts as the SNR grows, because at higher SNR the
noise floor is lower and the model bias appears earlier, consistent
with a second-order fit.

\section{Kernel Ablation: Edge Regime and Reduced Padding}
\label{supp:kernel}

\begin{figure*}[t]
  \centering
  \includegraphics[width=\textwidth]{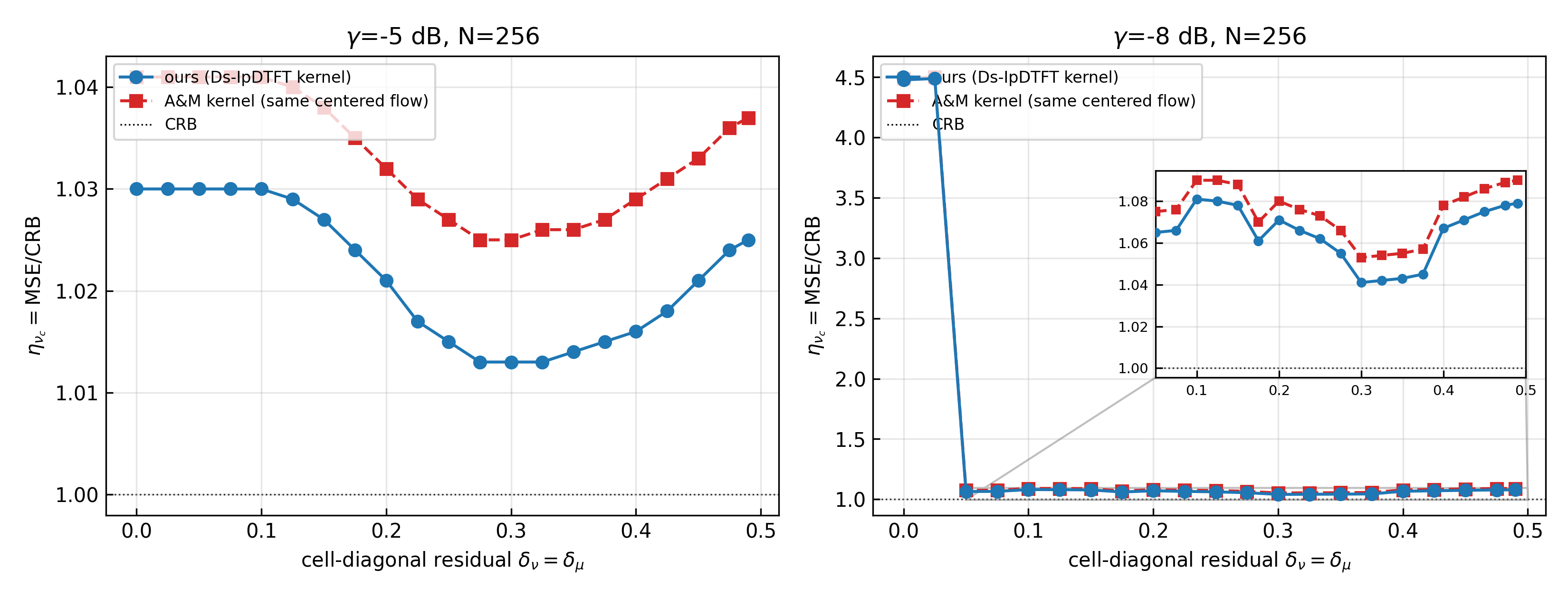}
  \caption{Kernel isolation at low SNR ($N=256$,
  MC $=2\times10^{4}$):
  identical centering, coarse stage, and chirp-rate stencil; only
  the frequency kernel differs (selectable-$p$, $p=0.35$, versus
  half-bin A\&M). Left: $\gamma=-5$~dB; right: $-8$~dB. The
  selectable-$p$ kernel is uniformly the better of the two
  (${\approx}1\%$ at the edge); under $M_{\mathrm{fac}}=3$
  padding, the refinement residual stays below $1/6$ bin, so the
  half-bin degradation of A\&M is largely suppressed and the remaining
  gap is small. Inset (right): zoom of $\delta\ge0.05$, past the
  outlier-inflated points at $\delta\le0.03$.}
  \label{fig:kernel}
\end{figure*}

Fig.~\ref{fig:kernel} isolates the contribution of the interpolation
kernel itself to the estimation accuracy, in the regime where the
two kernels differ: low SNR, with the residual near the
cell edge \cite{Wei2022Selectable,Wei2023DsIpDTFT}. With every other
component identical, the selectable-$p$ kernel tracks below the A\&M
kernel at all tested cell positions ($1.025$ versus $1.037$ at the
$-5$~dB corner; $1.079$ versus $1.090$ at $-8$~dB). \emph{Within the
padded pipeline} the gap is only ${\approx}1\%$, because the
zero-padded coarse stage already confines the refinement residual to
$|\delta|\le1/6$, where even the half-bin samples used in the A\&M
kernel remain inside
the main lobe; the structural advantage of $p<\tfrac12$
(Remark~\ref{rem:p}) is fully realized only at large controlled
residual; in the 1-D setting it reaches ${\approx}7\%$
\cite{Wei2023DsIpDTFT}.

\begin{figure*}[t]
  \centering
  \includegraphics[width=\textwidth]{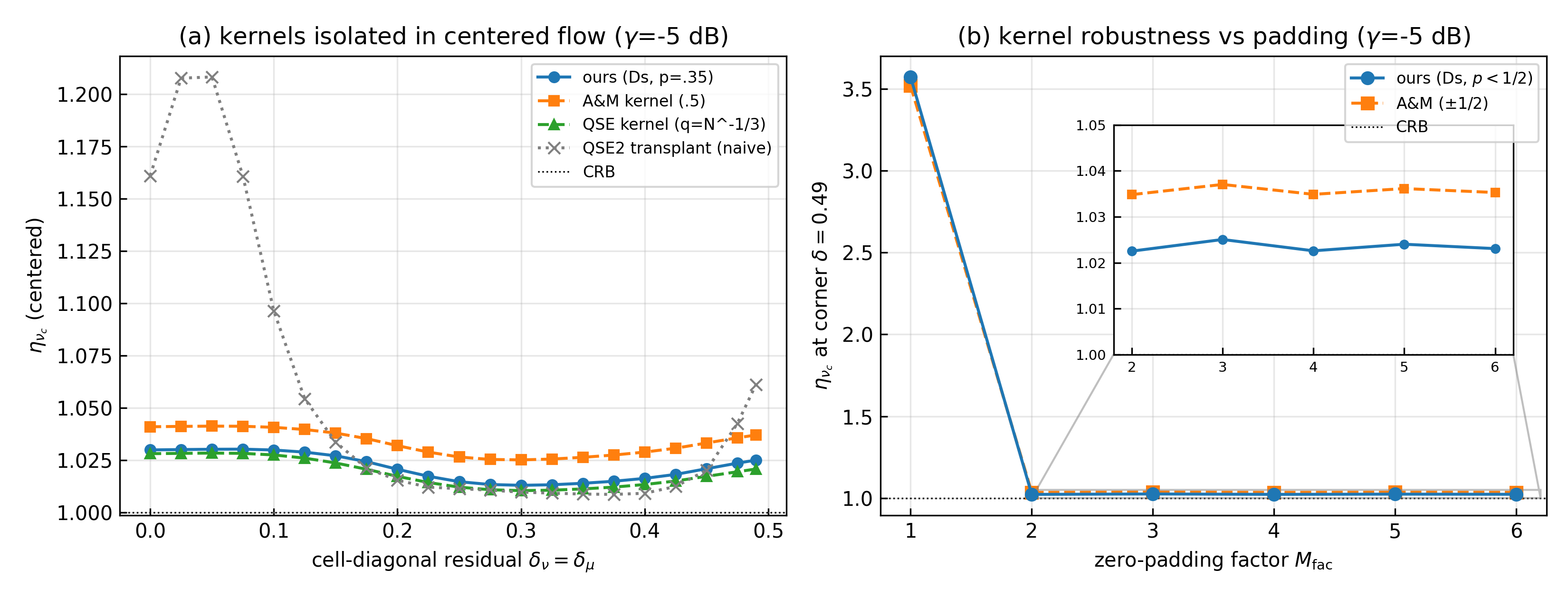}
  \caption{Two control experiments on the role of the frequency kernel
  ($N=256$, $\gamma=-5$~dB, centered $\eta_{\nu_{c}}$;
  MC $=2\times10^{4}$ per point).
  (a)~Kernels isolated in the identical centered, padded framework with the
  same chirp-rate stencil: the selectable-$p$, A\&M, and QSE kernels
  coincide and are uniform along the cell diagonal, whereas the naive
  QSE2 \emph{transplant} (uncentered, unpadded) varies nonmonotonically---%
  its deviation is thus an artifact of its structure, not a kernel property.
  (b)~Corner efficiency
  ($\delta=0.49$) versus padding factor: reducing $M_{\mathrm{fac}}$
  degrades both kernels together (coarse-stage scalloping loss, kernel
  independent), so the kernel choice
  does not compensate for reduced padding. Inset: zoom of
  $M_{\mathrm{fac}}=2$--$6$, where both kernels are flat and
  separated only by the ${\approx}1\%$ kernel-constant offset.}
  \label{fig:kernel-robust}
\end{figure*}

Two control experiments separate the kernel from the framework
(Fig.~\ref{fig:kernel-robust}). \emph{First}, when the three kernels
are placed in the identical centered, padded framework with the same
chirp-rate stencil (isolating the frequency kernel as for the
A\&M variant), the selectable-$p$, A\&M, and QSE kernels are
nearly indistinguishable and uniform along the cell diagonal
(Fig.~\ref{fig:kernel-robust}(a), all within $1.01$--$1.04$ at
$-5$~dB). The naive QSE2
\emph{transplant}
varies over the cell instead. A QSE
kernel embedded in the same framework is as uniform as ours. \emph{Second}, reducing the padding factor
(Fig.~\ref{fig:kernel-robust}(b)) is not compensated by either
kernel: at $M_{\mathrm{fac}}=1$ the corner
efficiency of both kernels degrades together to
${\approx}3.5$--$3.6$, and for
$M_{\mathrm{fac}}\ge2$ both sit at ${\approx}1.02$--$1.04$. The
$M_{\mathrm{fac}}=1$ degradation is driven by the coarse-stage
scalloping loss and the incipient outliers of
Theorem~\ref{thm:threshold}, which are kernel-independent and, at
this trial count, well sampled. The
selectable-$p$ kernel's benefit is the modest, controlled-residual edge
improvement above, and its selection here is primarily for continuity
with the estimator family of \cite{Wei2022Selectable,Wei2023DsIpDTFT}.

\section{Additional Verification Figures}

Fig.~S5 supports the discussion of the interpolation factor in
Section~III-B and Section~V-A of the main paper; Fig.~S6 reports the
efficiency constant of Theorem~\ref{thm:edgefree} and the noiseless
single-step inversion residual of Lemma~\ref{lem:inverse} as
functions of the signal length.

\begin{figure}[!t]
  \centering
  \includegraphics[width=\columnwidth]{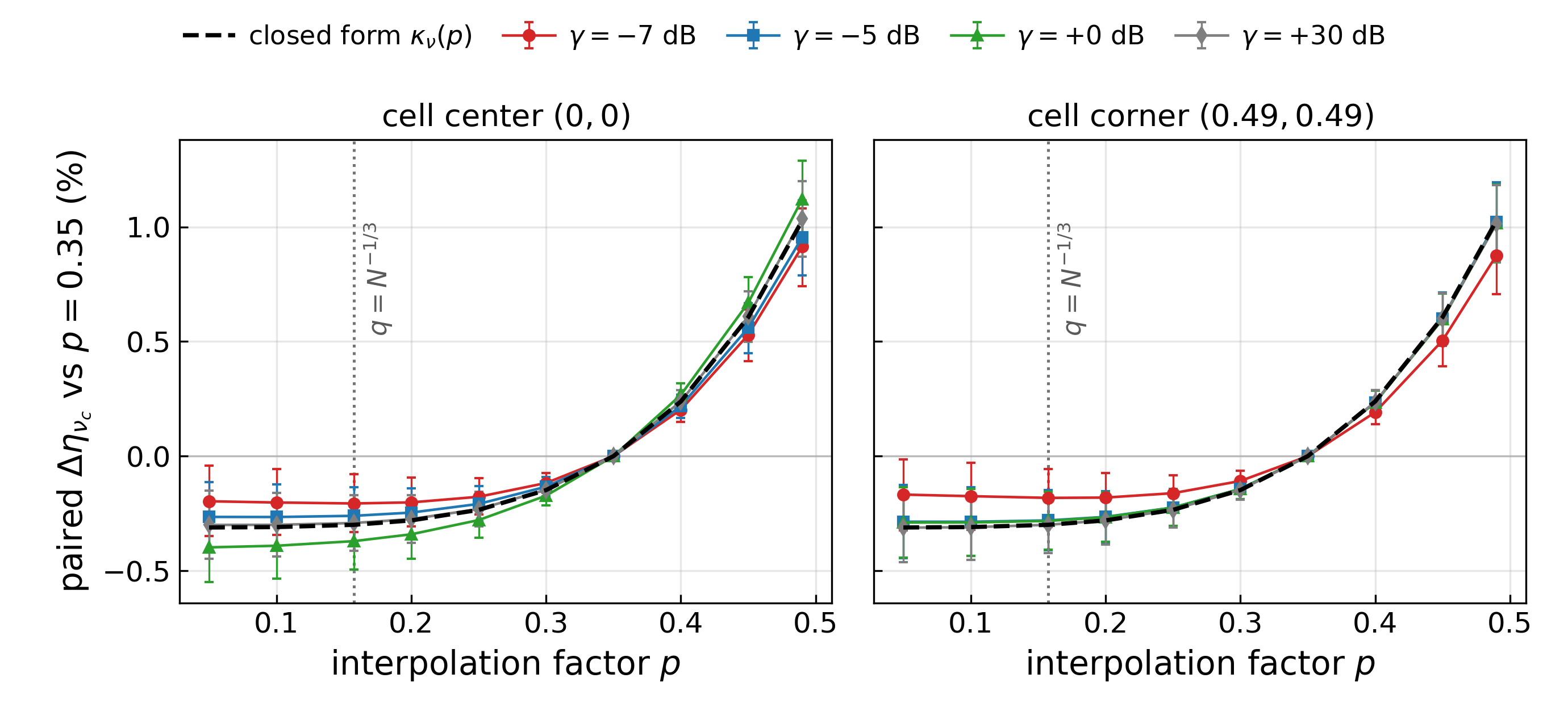}
  \caption{Paired change in $\eta_{\nu_{c}}$ relative to the fixed
  value $p=0.35$ over $p\in[0.05,0.49]$, at the cell center (left)
  and the $(0.49,0.49)$ corner (right); $N=256$,
  MC $=2\times10^{4}$ per point, common noise and the same coarse
  estimation stage for all $p$, bars are paired bootstrap $95\%$ CIs. The dashed
  curve is the closed form $\kappa_{\nu}(p)$ of \eqref{eq:kappa};
  the dotted line marks the QSE shift $q=N^{-1/3}$. The profiles
  track the closed form from $-7$~dB to $+30$~dB with a total
  spread of $1.0$--$1.5\%$; no trial reached
  the half-bin outlier threshold ($0$ of $2.4\times10^{5}$ over the
  three cell positions).}
  \label{fig:p-sweep}
\end{figure}

\begin{figure}[t]
  \centering
  \includegraphics[width=\columnwidth]{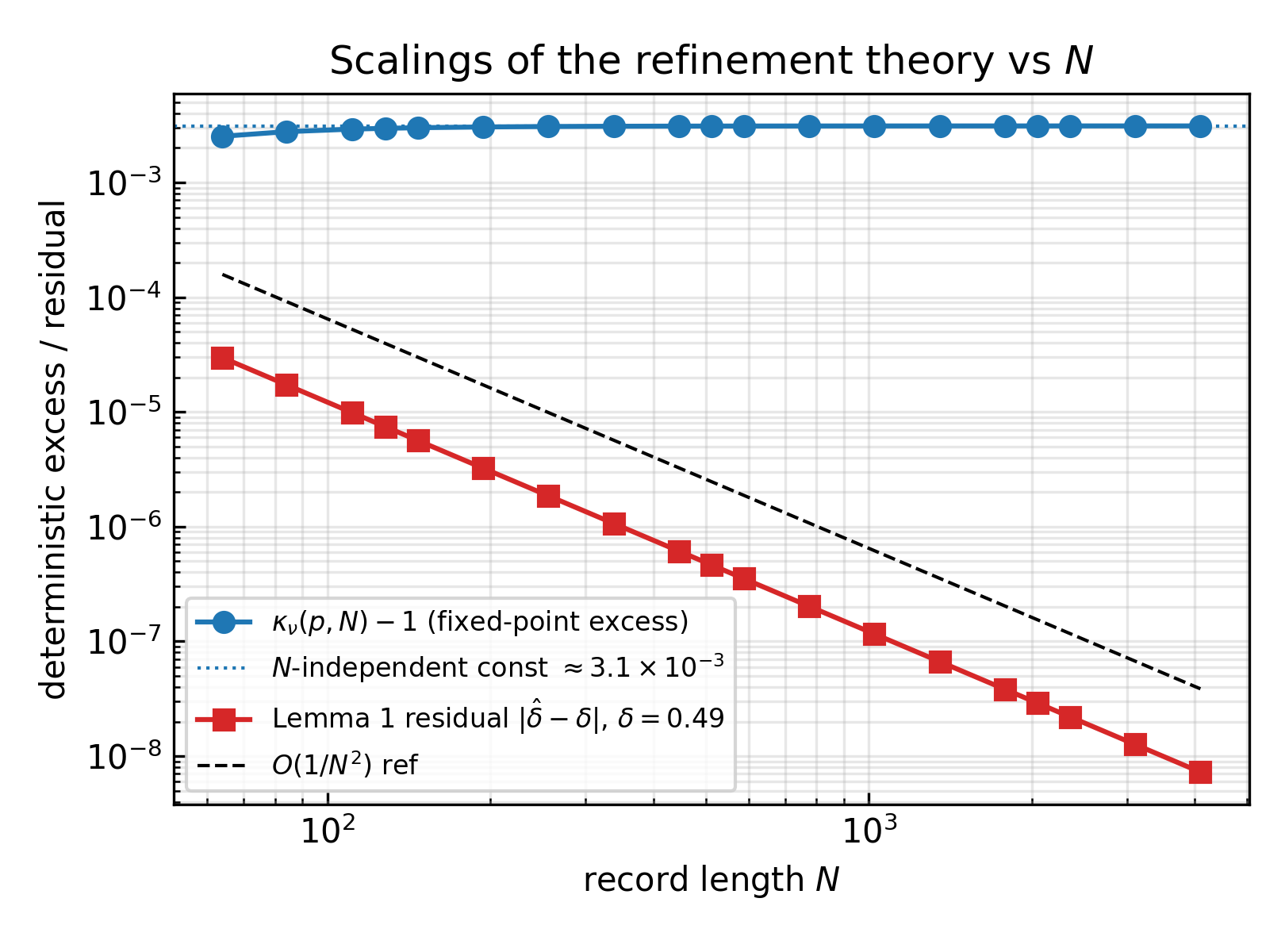}
  \caption{The efficiency constant $\kappa_{\nu}$ of
  Theorem~\ref{thm:edgefree} (closed form) and the noiseless
  single-step residual of Lemma~\ref{lem:inverse} (measured) versus
  the signal length $N$. The fixed-point efficiency excess
  $\kappa_{\nu}(p,N)-1$ tends to an $N$-independent constant
  (${\approx}3\times10^{-3}$ for $p=0.35$), so the
  estimator is CRB-efficient to a small fixed factor across signal
  lengths; the
  noiseless single-step inversion residual of Lemma~\ref{lem:inverse}
  at $\delta=0.49$ follows the $O(1/N^{2})$ reference over
  $N=64$--$4096$.}
  \label{fig:nscaling}
\end{figure}

\section{Additional Remarks}

\begin{remark}[Relation to the 1-D threshold theory]\label{rem:sq}
For a single tone, the exact acquisition probability was derived by
Serbes and Qaraqe \cite{SerbesQaraqe2022} as
$p=[1-\tfrac12 e^{-N\rho/2}]^{N-1}$, together with Lambert-$W$
closed forms for the breakdown and no-information thresholds,
building on the classical outlier analyses of
\cite{RifeBoorstyn1974,QuinnKootsookos1994}. Equation~\eqref{eq:pd} is the conditional
form of the same construction and reduces to a comparable expression
in the 1-D on-grid case ($L^{2}=1$, $\mathcal{M}\to N-1$).
Theorem~\ref{thm:threshold}
extends the theory in three directions specific to the joint problem:
the acquisition competes against a single effective count of
\emph{two-dimensional} cells; the threshold acquires a
\emph{continuous cell-position dependence} through the scalloping loss
\eqref{eq:straddle}, controllable by the padding factor; and the
floors $\Sigma_{\nu},\Sigma_{\mu}$ are per-axis prior variances. The
three-segment structure itself (prior floor / threshold / asymptotic)
is the classical interval-error picture, which underlies both
Ziv--Zakai-type bounds on any estimator \cite{Bell1997ZZB} and
estimator-specific interval-error threshold predictions for maximum
likelihood \cite{Athley2005Threshold};
Theorem~\ref{thm:threshold} provides a prediction of the latter kind,
in closed form, for the joint problem---with, specific to two
dimensions, the cell-position dependence through the scalloping loss. To
our knowledge no closed-form threshold characterization of joint
$(\nu,\mu)$ estimation is available in the literature; existing
chirp estimators---including the low-threshold quasi-maximum-likelihood
family \cite{Djurovic2014QML,Djurovic2018QMLReview}---are analyzed in
the asymptotic (CRB) region \cite{AldimashkiSerbes2020}, with
threshold behavior reported by simulation; the noisy-case analytical
models of \cite{AldimashkiSerbes2024} lower the breakdown
threshold, but do not characterize
the threshold location, its cell-position dependence, or the
three-segment MSE of the complete estimator.
\end{remark}

\begin{remark}[Necessity of centering]\label{rem:ablation}
Centering is a structural requirement. Without
it ($c=0$ in the dechirping), the residual chirp
$e^{j\pi\Delta\mu n^{2}}$, $n=0,\dots,N-1$, is asymmetric. Its linear
phase component displaces the apparent peak frequency by
$\approx\Delta\mu\,(N-1)/2$, injecting the chirp-rate error directly
into the frequency axis---this is the $-0.97$ correlation of
Section~\ref{ssec:crb} acting inside the refinement loop.
Section~\ref{sec:experiments} quantifies the resulting collapse for
all three kernels.
\end{remark}

\begin{remark}[Relation to 1-D edge-effect removals]\label{rem:morelli}
In one dimension, near-CRB accuracy at arbitrary residuals is
obtained without iteration by the WLS estimator of Morelli
\emph{et al.} \cite{MorelliWLS2022} and by the residual-adaptive
interpolator of D'Amico--Morelli \cite{DAmicoMorelli2022}, and the
selectable-$p$ kernel avoids the low-SNR edge effect
\cite{Wei2022Selectable,Wei2023DsIpDTFT}; all of these results
concern a single tone. Theorem~\ref{thm:edgefree} is, to our knowledge, the
first cell-uniform efficiency analysis for \emph{joint} $(\nu,\mu)$ interpolation,
where the two-dimensional residual cell (and its corners) has no 1-D
counterpart, and where centering---irrelevant in 1-D, where no
quadratic term exists---becomes the decisive ingredient.
\end{remark}

\appendices

\section{Proof of Theorem~\ref{thm:threshold}}
(i) Condition on the signal-cell power $s$. The noise-cell powers are
i.i.d.\ $\operatorname{Exp}(1)$, so the probability that all
$\mathcal{M}$ of them fall below $s$ is $(1-e^{-s})^{\mathcal{M}}$;
averaging over \eqref{eq:signal-cell-pdf} gives \eqref{eq:pd}. The
weak correlation of padded bins is absorbed into the calibrated
$\mathcal{M}$ (Remark~\ref{rem:thm1-honesty}). (ii) rests on two
approximations: conditioned on
correct acquisition the refinement is treated as CRB-accurate
(supported by Theorem~\ref{thm:edgefree} and the measurements of
Section~\ref{sec:experiments}), and conditioned
on an outlier the
selected cell is treated as uniform over the prior region, with
variance $\Sigma_{\theta}$---conservative near threshold, where
mis-selections are cell-scale near misses
(Remark~\ref{rem:thm1-honesty}); total expectation gives
\eqref{eq:three-segment}. (iii) For large $\mathcal{M}$,
$(1-e^{-s})^{\mathcal{M}}=\exp\bigl(\mathcal{M}\ln(1-e^{-s})\bigr)
\approx\exp(-\mathcal{M}e^{-s})$, a doubly exponential function of $s$
that switches from $0$ to $1$ in an $O(1)$ neighborhood of
$s=\ln\mathcal{M}$. Hence $P_{d}\approx\Pr\{s>\ln\mathcal{M}\}$, and
since $\mathbb{E}[s]=1+\Gamma$ concentrates around $\Gamma$ for large
$\Gamma$, the transition occurs at $\Gamma\approx\ln\mathcal{M}$;
substituting $\Gamma=NL^{2}\rho$ yields \eqref{eq:gth}. The effective
count $\mathcal{M}$ is calibrated once on a single configuration
(Remark~\ref{rem:thm1-honesty}).

\section{Proof of Lemma~\ref{lem:inverse}}
After exact dechirping, the noiseless data is a pure tone; its
DTFT magnitude at offset $r$ from the nearest bin is
$|X(\hat{k}+r)|=A\,|D(\delta-r)|$, $D(r)=\sin\pi r/\sin(\pi r/N)$.
For large $N$ the kernel tends to $D(r)\to N\operatorname{sinc}(r)
=N\sin\pi r/(\pi r)$; substituting the three samples
$r\in\{-p,0,p\}$ into \eqref{eq:ds-update} with this asymptotic form
and simplifying the resulting ratio yields $\hat\delta=\delta$
identically in $\delta$, the denominator constant $\cos\pi p$ being
the value that effects the cancellation
(Remark~\ref{rem:kernel}). For finite $N$ the Dirichlet kernel
differs from its $\operatorname{sinc}$ limit by $O(1/N^{2})$, so the
single-step inversion carries a residual of the same order; it is
\emph{uniform} in $\delta$.

\section{Proof of Theorem~\ref{thm:edgefree}}
The proof rests on two pillars.

\emph{Pillar A (centered decoupling).} By
\eqref{eq:fim-centered} the FIM of $(\phi_{c},\nu_{c},\mu)$ is
diagonal in the $(\nu_{c},\mu)$ block, and the refinement
inherits the same decoupling. A chirp-rate error $\Delta\mu$
leaves a residual chirp $e^{j\pi\Delta\mu(n-c)^{2}}$ that is even
about the signal center and hence does not displace the frequency
peak, while a frequency error enters the chirp-rate stencil only
through the locally flat peak magnitude. To leading order, the two
axes therefore
neither bias nor inflate each other, independently of the cell
position, and the
final linear map preserves efficiency since the centered estimates
are uncorrelated at that order; higher-order interactions enter the
$O(1/(N\rho))$ budget term.

\emph{Pillar B (kernel inversion and fixed point).} By
Lemma~\ref{lem:inverse} and the contraction property, the terminal
update operates at $\delta_{\mathrm{eff}}\approx0$. Linearizing
\eqref{eq:ds-update} there: the three magnitudes are
$|X(\hat{k})|\approx AN+u_{0}$ and
$|X(\hat{k}\pm p)|\approx A\,D(p)+u_{\pm}$, where
$u_{0}$ and $u_{\pm}$ are the in-phase noise components of the
corresponding DTFT samples: writing each sample as its signal term
plus complex noise $W$, $u=\operatorname{Re}\{W e^{-j\varphi}\}$ with
$\varphi$ the phase of the signal term. Each has
$\operatorname{var}(u)=N\sigma^{2}/2$, and
$\mathbb{E}[u_{+}u_{-}]=\tfrac{\sigma^{2}}{2}D(2p)$ (the DTFT noise
samples at spacing $2p$ are correlated through the Dirichlet kernel).
The signal terms cancel in the numerator of \eqref{eq:ds-update},
leaving
$\hat\delta\approx p\,(u_{+}-u_{-})/\bigl(2A[D(p)-N\cos\pi p]\bigr)$
with
$\operatorname{var}(u_{+}-u_{-})=\sigma^{2}\bigl(N-D(2p)\bigr)$;
converting bins to cycles/sample ($\nu_{c}=\hat{k}/N$) gives
\eqref{eq:kappa}. No step of the closed-form computation involves
$\delta_{\nu}$ or $\delta_{\mu}$, so $\kappa_{\nu}(p,N)$---the leading
term $\varepsilon_{0}$---is cell-position independent and tends
to an $N$-independent constant. The only residual dependence enters
through $\delta_{\mathrm{eff}}=O(\xi^{Q})$, contributing $O(\xi^{2Q})$,
and through the finite-SNR $O(1/(N\rho))$ term carried by the amplitude
nonlinearity of $|X|$; both are $\delta$-free bounds. At the corner, the
denominator $D(p)-N\cos\pi p$ remains a finite
nonzero constant, so
$\varepsilon_{\max}$ is finite. Within the linearization, every term of
the budget is nonnegative, so $\varepsilon_{\theta}\ge0$ to leading
order; Monte Carlo measurements scatter around $1+\varepsilon_{0}$ within
sampling error, on either side (Section~\ref{sec:experiments}).

\section{Proof of Proposition~\ref{prop:jerk}}
At high SNR a maximum-likelihood fit performs a least-squares
projection of the
unmodeled phase $2\pi\tfrac{\zeta}{6}n'^{3}$ onto the model basis
$\{1,n',n'^{2}\}$ spanned by the fitted parameters; the proposed
interpolation estimator tracks this projection up to a small
kernel-dependent factor, quantified in
Section~\ref{sec:experiments}. On the centered
grid all odd moments through the fifth vanish
($\sum n'=\sum n'^{3}=\sum n'^{5}=0$), so the cubic perturbation is
orthogonal to the even basis vectors $1$ and $n'^{2}$---hence zero
bias on $\phi_{c}$ and $\mu$---and projects only onto $n'$, with
coefficient $\tfrac{\zeta}{6}\sum n'^{4}/\sum n'^{2}
=\tfrac{\zeta}{6}\cdot\tfrac{3N^{2}-7}{20}$, giving
\eqref{eq:jerk-bias}. Equating $|\mathrm{bias}(\hat\nu_{c})|$ to
$\sqrt{\mathrm{CRB}_{\nu_{c}}}$ gives \eqref{eq:jerk-domain}.

\end{document}